\documentclass[journal,onecolumn,12pt]{IEEEtran}
\usepackage[numbers,sort,compress]{natbib}
\usepackage{geometry}
\usepackage{graphicx}
\usepackage{booktabs}
\usepackage{array}
\usepackage{longtable}
\usepackage{float}
\usepackage{microtype}

\usepackage{tikz}
\usetikzlibrary{calc}
\usepackage{enumitem}

\usepackage{amsfonts,color,morefloats}
\usepackage{amssymb,amsmath,latexsym,amsthm}
\usepackage{hyperref}
\hypersetup{colorlinks=true, linkcolor=blue, citecolor=blue, urlcolor=blue}
\usepackage{color}
\usepackage{makecell}
\usepackage{bm}
\usepackage{comment}

\makeatletter
\def\bstctlcite{\@ifnextchar[{\@bstctlcite}{\@bstctlcite[@auxout]}}
\def\@bstctlcite[#1]#2{\@bsphack
  \@for\@citeb:=#2\do{%
    \edef\@citeb{\expandafter\@firstofone\@citeb}%
    \if@filesw\immediate\write\csname #1\endcsname{\string\citation{\@citeb}}\fi}%
  \@esphack}
\makeatother

\newtheorem{theorem}{Theorem}
\newtheorem{lemma}{Lemma}
\newtheorem{proposition}{Proposition}
\newtheorem{corollary}{Corollary}

\newtheorem{definition}{Definition}
\newtheorem{example}{Example}

\newtheorem{remark}{Remark}

\newcommand{\F}{\mathbb{{F}}}

\newcommand{\mb}[1]{\mathbf{#1}}

\newcommand{\wt}{{\mathrm{wt}}}

\newcommand{\supp}{{\mathrm{supp}}}

\newcommand{\lrc}{{\mathrm{LRC}}}

\newcommand{\ud}{{\mathrm{d}}}
\newcommand{\emailaddr}[1]{\href{mailto:#1}{\nolinkurl{#1}}}

\title{Moment-based linear programming bounds for locally recoverable codes}
\author{Shujian Li\thanks{Department of Mathematics, The Hong Kong University of Science and Technology, Clear Water Bay, Hong Kong. Email: \emailaddr{slidv@connect.ust.hk}.},
\and Hengjia Wei\thanks{School of Mathematics and Statistics, Xi'an Jiaotong University, Xi'an 710049, China. Email: \emailaddr{hjwei05@xjtu.edu.cn}.},
\and Maosheng Xiong\thanks{Department of Mathematics, The Hong Kong University of Science and Technology, Clear Water Bay, Hong Kong. Email: \emailaddr{mamsxiong@ust.hk}.}}
\date{}

\begin{document}

\bstctlcite{IEEEexample:BSTcontrol}

\maketitle

\begin{abstract}
In this paper we derive new Delsarte-type linear programming bounds for
$q$-ary $(r,\delta)$-locally recoverable codes (LRCs) with three
attributes: first, the variable set is comparable in size to that of the
classical Delsarte LP; second, our LP exploits the higher-order information forced by the local-distance condition through order \(\delta-2\), in the sense that for nondegenerate linear codes, its balanced base part gives exactly the same dimension bound as the symmetrized
refined-weight LP of Gruica, Jany, and Ravagnani, while
the additional constraints, nonvacuous whenever $\delta \ge 3$, give a
further strengthening; and third, it applies to general
$(r,\delta)$-LRCs, linear and nonlinear alike. Extensive computations over binary and ternary alphabets show that the convex-hull LP yields improvements not captured by the previous LP and often sharpens the shortening and generalized Singleton bounds. 
\end{abstract}


\section{Introduction}
An \((r,\delta)\)-locally recoverable code (LRC) is a code in which
every coordinate lies in a local set of size at most \(r+\delta-1\)
whose restriction has minimum distance at
least~\(\delta\)~\cite{prakashOptimalLinear2012,freij-hollantiMatroidTheory2018}.
Consequently, any \(\delta-1\) erasures inside such a set can be
repaired from the set alone.  The locality paradigm arose from the need
to limit the number of helper nodes and disk accesses involved in
repairing failed nodes in distributed
storage~\cite{gopalanLocalityCodeword2012,huangErasureCodingAzure2012}:
the case \(\delta=2\) is the original notion of locality~\(r\), and
general \(\delta\) is the refinement allowing several simultaneous
failures within one group.

It has since developed into a broad coding-theoretic program
encompassing
constructions~\cite{tamoFamilyOptimal2014,bargLocallyRecoverableCurves2017,micheliConstructionsOptimal2020,haoBoundsConstructions2020},
alphabet-dependent converse
bounds~\cite{cadambeBoundsSize2015,agarwalCombinatorialAlphabetdependent2018,grezetAlphabetDependentBounds2019a},
bounds obtained by matroid and polymatroid
methods~\cite{tamoOptimalMatroid2016,freij-hollantiMatroidTheory2018,westerbackPolymatroid2015}, codes that are not assumed to be
linear~\cite{papailiopoulosLocallyRepairable2014,forbesLocalityCodewordSymbols2014},
limits on the length of optimal codes~\cite{guruswamiHowLong2019}, and
connections with
regeneration~\cite{dimakisNetworkCodingDSS2010,kamathLocalRegeneration2014,balajiErasureCoding2018}. 
Determining the exact trade-off among the block length \(n\), the code
size \(|C|\), the global minimum distance \(d\), the locality \(r\),
the local distance \(\delta\), and the alphabet size \(q\) remains a
fundamental open problem.

Two bounds serve as the standard baselines.  The generalized Singleton
bound 
in~\cite{gopalanLocalityCodeword2012,prakashOptimalLinear2012}, 
with extensions to arbitrary codes through entropy polymatroids
~\cite{westerbackPolymatroid2015,freij-hollantiMatroidTheory2018}, is explicit and broadly applicable
but, being alphabet-independent, can be loose over small fields.  In
that regime the sharpest known bounds often come from shortening
arguments~\cite{cadambeBoundsSize2015,grezetAlphabetDependentBounds2019a},
whose evaluation, however, depends on the largely unknown
classical-code functions \(A_q(N,d)\) and
\(k_{\mathrm{opt}}^{(q)}(N,d)\) across many shortened lengths \(N\),
and must ultimately rely on further bounds or on
tables~\cite{grasslCodeTables,brouwerCodeTables}.

Linear programming (LP), pioneered by Delsarte~\cite{Delsarteassoc}, 
is among the most powerful tools for bounding codes over a fixed 
alphabet; see for example the McEliece--Rodemich--Rumsey--Welch (MRRW) rate bounds~\cite{mcelieceNewUpper1977}.
Yet incorporating locality into an LP is considerably more 
subtle. An earlier product-association-scheme LP of Agarwal, Barg, Hu,
Mazumdar, and Tamo applies to arbitrary codes, but assumes that the coordinates are partitioned into
disjoint repair groups of a common size
~\cite{agarwalCombinatorialAlphabetdependent2018}. More recently, Gruica, Jany, and Ravagnani~\cite{gruicaDualityLP2023,gruicaLRCSDuality2026} introduced a refined weight distribution for nondegenerate
linear \((r,\delta)\)-LRCs and established a MacWilliams-type
identity leading to a locality-aware LP.

Three limitations of the refined-weight LP motivate the present work. First, for \(\delta>2\), its low-weight dual locality constraint is necessary but does not exploit the higher local-distance information; see~\cite[Remark~4.12]{gruicaLRCSDuality2026}. Second, its variables are 
indexed by both weights and coordinates, making the formulation substantially larger than 
the classical Delsarte LP. Third, the method is inherently 
linear and does not extend to nonlinear LRCs.

\smallskip
\noindent\textbf{Our contribution.}
We address these three limitations by deriving local factorial-moment identities through order \(\delta-2\). Averaging these identities over the recovered coordinates produces a compact Delsarte-type LP with a number of variables comparable to that of the classical distance-distribution LP. The construction applies to arbitrary linear and nonlinear \((r,\delta)\)-LRCs with general recovery sets. We obtain a box LP, which treats the moments separately, and a convex-hull LP, which preserves their joint feasibility. For nondegenerate linear codes, the balanced base LP gives the same dimension bound as the symmetrized refined-weight LP
of~\cite{gruicaLRCSDuality2026}.  When \(\delta=2\), the positive-order moment constraints are vacuous; for
\(\delta\ge3\), they provide the new strengthening. Extensive computations over binary and ternary alphabets show that the convex-hull LP yields improvements not captured by the previous LP and often sharpens the shortening and generalized Singleton bounds.

The paper is organized as follows. Section~\ref{sec:prel} fixes notation and reviews the background
used throughout: the definition of a general \((r,\delta)\)-LRC, the generalized Singleton and shortening bounds, the Krawtchouk polynomials, and the MacWilliams transform. Section~\ref{sec:facmomid} establishes the local factorial-moment identities, first for linear LRCs and then in the nonlinear
setting. Section~\ref{sec:main} derives the two resulting bounds, the
box LP and the convex-hull LP. Section~\ref{sec:comparison} shows that, for nondegenerate linear codes, our zeroth-moment LP is
equivalent after symmetrization to the refined-weight LP of~\cite{gruicaLRCSDuality2026}, yielding a compact reformulation of the latter. Section~\ref{sec:computations} presents numerical comparisons,
and Section~\ref{sec:conclusion} concludes.

\section{Preliminaries}\label{sec:prel}
Throughout the paper, we fix the following notations. 

Let \(q\) be a prime power and \(\F_q\) the finite field of order
\(q\). For a positive integer \(n\), write \([n] \triangleq
\{1,2,\ldots,n\}\). For \(S \subseteq [n]\), write \(S^c
\triangleq [n] \setminus S\), denote its cardinality by \(|S|\),
and $\F_q^S$ denotes the vector space over $\F_q$ whose coordinates are indexed by $S$. We identify $\F_q^n$ with $\F_q^{[n]}$. Let \(\pi_S : \F_q^n \to \F_q^{S}\) be the coordinate
projection onto \(S\). For \(\mb x = (x_1, \ldots, x_n) \in
\F_q^n\), the \emph{Hamming weight} and \emph{support} of \(\mb x\)
are
\[
\wt(\mb x) \triangleq |\{i : x_i \ne 0\}|, \qquad
\supp(\mb x) \triangleq \{i : x_i \ne 0\}.
\]
The \emph{Hamming distance} between \(\mb x, \mb y \in \F_q^n\) is
\(\ud(\mb x, \mb y) \triangleq \wt(\mb x - \mb y)\). The
\emph{minimum distance} of a code \(C \subseteq \F_q^n\) is
\[
\ud(C) \triangleq \min_{\mb x \ne \mb y \in C} \ud(\mb x, \mb y),
\]
with the convention \(\ud(C) = 0\) when \(|C| \le 1\).

\subsection{LRCs and baseline bounds}

\begin{definition}[$(r,\delta)$-LRC]\label{def:lrc-general}
	Let \(r\) and \(\delta\) be positive integers with \(\delta \ge
	2\). A code \(C \subseteq \F_q^n\), not necessarily linear, is an
	\emph{\((r,\delta)\)-locally recoverable code} (\((r,\delta)\)-$\mathrm{LRC}$)
	if for every \(i \in [n]\) there exists a subset \(S_i \subseteq
	[n] \setminus \{i\}\) such that
	\[
	|S_i| \le r + \delta - 2 \quad \text{and} \quad
	\ud\bigl(\pi_{S_i \cup \{i\}}(C)\bigr) \ge \delta.
	\]
	We call \(S_i\) an \emph{\((r,\delta)\)-recovery set} for
	coordinate \(i\).
\end{definition}

This definition, adopted from~\cite{freij-hollantiMatroidTheory2018},
covers both linear and nonlinear codes, in line with the
nonlinear-locality viewpoint of Forbes and
Yekhanin~\cite{forbesLocalityCodewordSymbols2014} for $\delta=2$. When \(C\) is
linear and \(\delta \ge 2\), it reduces to the standard linear
\((r,\delta)\)-LRC definition of Prakash et
al.~\cite{prakashOptimalLinear2012,gruicaLRCSDuality2026}.

The following simple property of $(r,\delta)$-LRCs might be easily overlooked: 
\begin{proposition}\label{prop:d-ge-delta} Let $r,\delta$ be positive integers with $\delta \ge 2$. If \(C \subseteq \F_q^n\) is an \((r,\delta)\)-LRC with at least two codewords, then
\begin{eqnarray*} 
\ud (C) \ge \delta.
\end{eqnarray*}
\end{proposition}
\begin{proof} For any distinct $\mb x, \mb y \in C$, choose a coordinate $i$ such that $x_i \ne y_i$. Let $S_i$ be a recovery set for $i$ and let $T_i=S_i \cup \{i\}$. Then $\pi_{T_i}(\mb x),\pi_{T_i}(\mb y)$ are distinct codewords of $\pi_{T_i}(C)$ of distance at least $\delta$. Thus $\ud(\mb x, \mb y) \ge \delta$. 
\end{proof}

We now recall two baseline upper bounds: the generalized Singleton bound and the shortening bound, against which we will compare our LP bounds in Section~\ref{sec:computations}. 

Let \(C \subseteq \F_q^n\) be an \((r,\delta)\)-LRC with minimum distance \(d\) where $n \ge d \ge \delta \ge 2$. Set
\[ K \triangleq \log_q |C|.\]
The \emph{generalized Singleton bound} states that
\begin{align}\label{eq:gen-singleton}
	d \le n - \lceil K \rceil + 1 - \left(\bigl\lceil \lceil K \rceil/r \bigr\rceil - 1\right)(\delta - 1).
\end{align}
For linear codes $K=k$ is an integer and \eqref{eq:gen-singleton} is the
bound of Prakash et al.~\cite{prakashOptimalLinear2012}. For general,
possibly nonlinear codes, \eqref{eq:gen-singleton} is the case $t=1$ of
the Singleton-type bound for polymatroids of Westerb\"ack, Freij and
Hollanti~\cite{westerbackPolymatroid2015}; a statement in the present
notation is~\cite[Theorem~10]{freij-hollantiMatroidTheory2018}. We let \(k_{\mathrm{GS}}\) denote the largest integer satisfying
\eqref{eq:gen-singleton}, so that
\[ K \,\le\, \lceil K \rceil \,\le\, k_{\mathrm{GS}} \]
for linear and nonlinear codes alike. 

For \(N\ge0\), \(A_q(N,d)\) denotes the largest cardinality
of a set \(D \subseteq \F_q^N\) such that 
\[\ud(\mb x,\mb y) \ge d\]
for all distinct \(\mb x,\mb y \in D\). Since \(d\ge2\), we have \(A_q(N,d)=1\) whenever
\(N<d\). 

\begin{theorem}[Shortening bound]
\label{thm:nonlinear-shortening}
Let \(C\subseteq \F_q^n\) be an \((r,\delta)\)-LRC with minimum
distance \(d\), and set \(K=\log_q |C|\). Then
\begin{align*}
K
\le
\min_{\ell\in\mathcal{L}}
\left\{
\ell r+
\log_q A_q\bigl(n-\ell R_{\mathrm{loc}},d\bigr)
\right\},
\end{align*}
where
\[
R_{\mathrm{loc}}\triangleq r+\delta-1,
\qquad
\mathcal L
\triangleq
\left\{
0,1,\ldots,
\left\lfloor
\frac{n}{R_{\mathrm{loc}}}
\right\rfloor
\right\}. 
\]
\end{theorem}

The alphabet-dependent shortening method was introduced by
Cadambe and Mazumdar for ordinary locality, corresponding to
\(\delta=2\), and their result applies to arbitrary, possibly
nonlinear, codes~\cite[Theorem~1]{cadambeBoundsSize2015}.
Rawat, Mazumdar, and Vishwanath obtained an analogous bound
for possibly nonlinear codes with cooperative locality
~\cite[Remark~3]{rawatCooperativeLocal2015a}, while Grezet
et al.\ established the linear all-symbol \((r,\delta)\)-version
~\cite[Corollary~2]{grezetAlphabetDependentBounds2019a}.

The cooperative-locality result does not directly imply
Theorem~\ref{thm:nonlinear-shortening}.  Indeed, in the
definition of Rawat et al., every set of \(\ell\) coordinates
must be jointly recoverable from at most \(R\) other
coordinates.  An \((r,\delta)\)-LRC generally guarantees only
\(((\delta-1)r,\delta-1)\)-cooperative locality, rather than
\((r,\delta-1)\)-cooperative locality.  Substitution into the
cooperative-locality bound therefore gives a different
inequality.  To the best of our knowledge, the arbitrary,
possibly nonlinear, all-symbol \((r,\delta)\)-statement in
Theorem~\ref{thm:nonlinear-shortening} has not previously been
stated explicitly.  Since the proof is very similar to those in \cite{cadambeBoundsSize2015,rawatCooperativeLocal2015a,grezetAlphabetDependentBounds2019a} and for the sake of completeness, we include a self-contained proof. 

\begin{proof}
For \(J\subseteq[n]\), put
\[
	h(J)\triangleq\log_q|\pi_J(C)|.
\]
Suppose that
\(h(J)<K\). Then \(|\pi_J(C)|<|C|\), so there are distinct
\(\mb x,\mb y\in C\) that agree on \(J\). Choose
\(i\notin J\) such that \(x_i\ne y_i\), let \(S_i\) be a recovery
set for \(i\), and write
\[
	T_i\triangleq S_i\cup\{i\},
	\qquad
	B\triangleq T_i\setminus J.
\]
The words \(\pi_{T_i}(\mb x)\) and
\(\pi_{T_i}(\mb y)\) are distinct and agree on \(T_i\cap J\).
Since \(\ud(\pi_{T_i}(C))\ge\delta\), they differ in at least
\(\delta\) coordinates of \(B\); in particular, \(|B|\ge\delta\).

Choose \(E\subseteq B\) with \(|E|=\delta-1\). The projection onto \(T_i\setminus E\) is injective on the
local code \(\pi_{T_i}(C)\), since otherwise two distinct local codewords
would differ in at most \(\delta-1\) coordinates. Consequently,
the map
\[
	\pi_{J\cup T_i}(C)
	\longrightarrow
	\pi_J(C)\times\F_q^{B\setminus E},
	\qquad
	\pi_{J\cup T_i}(\mb c)
	\longmapsto
	\bigl(\pi_J(\mb c),\pi_{B\setminus E}(\mb c)\bigr),
\]
is injective. Indeed, its two components determine the coordinates
in \(T_i\setminus E\), and injectivity of the punctured local code 
determines the coordinates in \(E\). It follows that
\begin{align}\label{eq:shortening-local-step}
	h(J\cup T_i)
	\le h(J)+|B|-(\delta-1).
\end{align}

Fix \(\ell\in\mathcal{L}\). If \(\ell r\ge K\), then 
\[
	K\le \ell r+
	\log_q A_q\bigl(n-\ell R_{\mathrm{loc}},d\bigr)
\]
is immediate since \(A_q(N,d)\ge1\). Now suppose \(\ell r<K\). Starting with \(J_0=\varnothing\), we apply the
local step recursively to construct
\[
	J_0\subseteq J_1\subseteq\cdots\subseteq J_\ell
\]
such that, for \(0\le j\le\ell\),
\begin{align}\label{eq:shortening-invariants}
	|J_j|\le jR_{\mathrm{loc}},
	\qquad
	h(J_j)\le |J_j|-j(\delta-1)\le jr.
\end{align}
Indeed, if \(j<\ell\), then \(h(J_j)\le jr\le\ell r<K\), so the
local step applies. Setting \(J_{j+1}=J_j\cup T_i\), with
\(B=T_i\setminus J_j\), gives
\[
	|J_{j+1}|=|J_j|+|B|\le(j+1)R_{\mathrm{loc}},
\]
and~\eqref{eq:shortening-local-step} gives
\[
	h(J_{j+1})
	\le |J_{j+1}|-(j+1)(\delta-1).
\]
This proves~\eqref{eq:shortening-invariants}.

Since \(\ell R_{\mathrm{loc}}\le n\), enlarge \(J_\ell\), if
necessary, to a set \(I\subseteq[n]\) with
\(|I|=\ell R_{\mathrm{loc}}\). Adding one coordinate increases
the logarithm of the projection size by at most one, and hence
\begin{align}\label{eq:shortening-projection}
	h(I)
	&\le h(J_\ell)+|I|-|J_\ell| \notag\\
	&\le |J_\ell|-\ell(\delta-1)
		+\ell R_{\mathrm{loc}}-|J_\ell|
	 =\ell r.
\end{align}

Choose a largest fiber of the projection onto \(I\); namely, choose
\(\mb a\in\pi_I(C)\) for which
\[
	C_{\mb a}\triangleq
	\{\mb c\in C:\pi_I(\mb c)=\mb a\}
\]
has maximum cardinality. By~\eqref{eq:shortening-projection},
\[
	|C_{\mb a}|
	\ge\frac{|C|}{|\pi_I(C)|}
	=q^{K-h(I)}
	\ge q^{K-\ell r}.
\]
All words in \(C_{\mb a}\) agree on \(I\). Therefore puncturing
\(I\) is injective on this fiber and preserves the distance
between every two of its distinct words. Thus
\(\pi_{I^c}(C_{\mb a})\) is a code of length
\(n-\ell R_{\mathrm{loc}}\), cardinality \(|C_{\mb a}|\), and
pairwise distances at least \(d\). By the definition of \(A_q\),
\[
	q^{K-\ell r}
	\le |C_{\mb a}|
	\le A_q\bigl(n-\ell R_{\mathrm{loc}},d\bigr).
\]
Taking logarithms proves the desired inequality for the fixed
\(\ell\). Minimizing over \(\ell\in\mathcal{L}\) completes the
proof.
\end{proof}

For linear codes, the following stronger version was established
by Grezet et al.~\cite[Corollary~2]{grezetAlphabetDependentBounds2019a}.

\begin{corollary}[Linear shortening bound]
\label{cor:linear-shortening}
Let \(C\le \F_q^n\) be a linear \((r,\delta)\)-LRC of
dimension \(k\) and minimum distance \(d\). Then
\[
k
\le
\min_{\ell\in\mathcal L}
\left\{
\ell r+
k_{\mathrm{opt}}^{(q)}
\bigl(n-\ell R_{\mathrm{loc}},d\bigr)
\right\},
\]
where \(k_{\mathrm{opt}}^{(q)}(N,d)\) denotes the largest
dimension of a linear \(q\)-ary code of length \(N\) and
minimum distance at least \(d\).
\end{corollary}

Let \(U_{\mathrm{Del}}^{(q)}(N,d)\) denote the ordinary
Delsarte LP upper bound on \(A_q(N,d)\), and define
\[
K_{\mathrm{LP}}^{(q)}(N,d)
\triangleq
\left\lfloor
\log_q U_{\mathrm{Del}}^{(q)}(N,d)
\right\rfloor.
\]
Corollary~\ref{cor:linear-shortening} gives the two linear shortening
benchmarks
\begin{align*}
k_{\mathrm{SH,LP}}
&\triangleq
\min_{\ell\in\mathcal L}
\left\{
\ell r+
K_{\mathrm{LP}}^{(q)}
\bigl(n-\ell R_{\mathrm{loc}},d\bigr)
\right\},
\\
k_{\mathrm{SH,opt}}
&\triangleq
\min_{\ell\in\mathcal L}
\left\{
\ell r+
k_{\mathrm{opt}}^{(q)}
\bigl(n-\ell R_{\mathrm{loc}},d\bigr)
\right\}.
\end{align*}
When all the required values of
\(k_{\mathrm{opt}}^{(q)}(N,d)\) are known exactly using Grassl's tables
~\cite{grasslCodeTables}, we evaluate the latter benchmark and denote it by
\(k_{\mathrm{SH,exact}}\).

For arbitrary codes, Theorem~
\ref{thm:nonlinear-shortening} gives the analogous benchmarks
\begin{align*}
M_{\mathrm{SH,LP}}
&\triangleq
\min_{\ell\in\mathcal L}
\left\{
\ell r+
\log_q U_{\mathrm{Del}}^{(q)}
\bigl(n-\ell R_{\mathrm{loc}},d\bigr)
\right\},
\\
M_{\mathrm{SH,opt}}
&\triangleq
\min_{\ell\in\mathcal L}
\left\{
\ell r+
\log_q A_q
\bigl(n-\ell R_{\mathrm{loc}},d\bigr)
\right\}.
\end{align*}
When all the required values of \(A_q(N,d)\) are known
exactly using Brouwer's binary and ternary code tables
~\cite{brouwerCodeTables}, we evaluate \(M_{\mathrm{SH,opt}}\) and denote it by
\(M_{\mathrm{SH,exact}}\).

Consequently,
\[
k\le k_{\mathrm{SH,opt}}\le k_{\mathrm{SH,LP}}, \quad \mbox{ if } C \mbox{ is linear}, 
\]
and
\[
K\le M_{\mathrm{SH,opt}}\le M_{\mathrm{SH,LP}}, \quad \mbox{ for arbitrary } C. 
\]

\subsection{Krawtchouk polynomials and MacWilliams transform}
For any positive integer $n$, the $q$-ary Krawtchouk polynomials $K_\bullet(\bullet;n,q)$ on length $n$ are defined as 
\begin{align*}
	K_\omega(j; n, q) \triangleq
	\sum_{\ell=0}^{\omega} (-1)^\ell (q-1)^{\omega - \ell}
	\binom{j}{\ell} \binom{n - j}{\omega - \ell}, \quad \forall \, 0 \le \omega,j \le n.
\end{align*}
Here \(\binom{a}{b}\) denotes the standard binomial coefficient
for nonnegative integers \(0 \le b \le a\), extended to all
integers \(a, b\) by the convention
\[
\binom{a}{b} = 0 \quad \text{if } a < 0,\ b < 0,\ \text{or } b > a.
\]
Equivalently, these Krawtchouk polynomials are given by the
generating function
\begin{align}\label{eq:Kraw-gen-prelim}
	\sum_{\omega=0}^n K_\omega(j; n, q)\, z^\omega
	= \bigl(1 + (q-1)z\bigr)^{n-j} (1 - z)^j.
\end{align}
Setting \(z = 1\) in~\eqref{eq:Kraw-gen-prelim} yields the
elementary identity
\begin{align}\label{eq:krawtchouk-sum}
	\sum_{\omega=0}^n K_\omega(j; n, q) =
		q^n\mathbf 1_{\{j=0\}}. 
\end{align}
We will also need the following elementary result: 
\begin{lemma}\label{lem:krawtchouk-moment}
	For all integers $0 \le j, m \le n$,
	\[
	\sum_{\omega = 0}^n \binom{\omega}{m} K_\omega(j; n, q) \,=\,
		(-1)^j\, q^{n-m}(q-1)^{m-j}\binom{n-j}{m-j}\mathbf 1_{\{j \le m\}}. 
	\]
\end{lemma}

\begin{proof}
	By \eqref{eq:Kraw-gen-prelim},  we have
	\begin{align} \label{2:leib}
		\sum_{\omega=0}^n \binom{\omega}{m} K_\omega(j; n, q)\, z^{\omega - m} \,=\, \frac{1}{m!}\frac{\ud^m}{\ud z^m}\!\left[(1 + (q-1)z)^{n-j}(1-z)^j\right].
	\end{align}
	We evaluate both sides of~\eqref{2:leib} at $z = 1$.  The left-hand side is $\sum_\omega \binom{\omega}{m} K_\omega(j; n, q)$. By Leibniz's rule, the right-hand side is 
	\begin{align*}
		\mbox{RHS}=\frac{1}{m!}\sum_{s=0}^m\binom{m}{s}\frac{\ud^{m-s}}{\ud z^{m-s}}\!\left[\left(1+(q-1)z\right)^{n-j}\right] \cdot \frac{\ud^{s}}{\ud z^{s}}\!\left[(1-z)^{j}\right] \\
		=\sum_{s=0}^m\binom{n-j}{m-s} (q-1)^{m-s}\left(1+(q-1)z\right)^{n-j-m+s}\binom{j}{s}(-1)^s  \left(1-z\right)^{j-s}. 
	\end{align*}
	If $j > m$, then $(1-z)^{j-s}$ vanishes at $z=1$ for every $0\le s \le m$.  If $0 \le j \le m$, evaluating at $z=1$ leaves only the term $s=j$, and therefore 
	\begin{align*}
		\mbox{RHS}&= \binom{n-j}{m-j}(q-1)^{m-j}\left[\left(1+(q-1)z\right)^{n-m}\right]_{z=1} \cdot (-1)^j \\
		&= (-1)^j q^{n-m}(q-1)^{m-j}\binom{n-j}{m-j}. \qedhere
	\end{align*}
	This completes the proof of Lemma \ref{lem:krawtchouk-moment}. 
\end{proof}
We will also use the standard Krawtchouk inversion identity
\begin{align*}
\sum_{j=0}^n
K_j(x;n,q)K_\omega(j;n,q)
=
q^n\mathbf 1_{\{x=\omega\}},
\qquad
0\le x,\omega\le n.
\end{align*}
Equivalently, if
\[
Y_\omega=\sum_{j=0}^nX_jK_\omega(j;n,q),
\]
then
\[
X_j=\frac1{q^n}
\sum_{\omega=0}^nY_\omega K_j(\omega;n,q).
\]

The \emph{distance distribution} \(\{A_j\}_{j=0}^n\) of a code \(C
\subseteq \F_q^n\) is defined by~(see \cite{MacWilliamsd})
\begin{align}\label{eq:Aj-def}
	A_j \triangleq \frac{1}{|C|}
	\left|\{(\mb c, \mb c') \in C^2 : \ud(\mb c, \mb c') = j\}\right|,
	\qquad 0 \le j \le n,
\end{align}
and satisfies \(A_0 = 1\) and \(\sum_{j=0}^n A_j = |C|\).

The \emph{MacWilliams transform} of \(\{A_j\}_{j=0}^n\) on length $n$ is the
sequence \(\{B_\omega\}_{\omega=0}^n\) given by
\begin{align}\label{eq:Bw-def}
	B_\omega \triangleq \frac{1}{|C|}
	\sum_{j=0}^n A_j K_\omega(j; n, q), \qquad 0 \le \omega \le n.
\end{align}
It satisfies \(B_0 = 1\) and \(B_\omega \ge 0\) for all
\(\omega\). The nonnegativity follows from the character-sum
representation
\begin{align*}
	B_\omega = \frac{1}{|C|^2}
	\sum_{\substack{\mb v \in \F_q^n \\ \wt(\mb v) = \omega}}
	\left|\sum_{\mb c \in C} \chi(\mb v \cdot \mb c)\right|^2,
\end{align*}
where \(\chi\) is any nontrivial additive character of \(\F_q\)
and \(\mb v \cdot \mb x \triangleq \sum_{i=1}^n v_i x_i\).

When \(C\) is linear, the quantities \(A_j\) coincide with the usual
\emph{weight distribution} of $C$
\[
W_j(C) \triangleq |\{\mb c \in C : \wt(\mb c) = j\}|,
\]
and \(\{B_\omega\}\) is then the weight distribution of the dual
code 
\[
C^\perp \triangleq \left\{\mb a \in \F_q^n : \mb a \cdot \mb c = 0
\text{ for all } \mb c \in C\right\}.
\]

\section{General factorial moment identities}\label{sec:facmomid}
Let $C$ be an $(r,\delta)$-$\mathrm{LRC}$ over $\F_q$ with minimum distance $d$, where $n,r,\delta,d$ are fixed positive integers such that $n \ge d \ge \delta \ge 2$. For each coordinate $i\in[n]$, we fix one recovery set $S_i$ and write
\[
T_i \triangleq S_i\cup\{i\},\qquad t_i\triangleq |T_i|,
\qquad R\triangleq \min\{n,R_{\mathrm{loc}}\}=\min\{n,r+\delta-1\}.
\]
Obviously 
\begin{equation*}
	2\le \delta\le t_i\le R.
\end{equation*}
Indeed, $\ud(\pi_{T_i}(C))\ge\delta$ implies that the local projection has length at least $\delta$.

\subsection{The linear case}\label{sec:linear}

We first treat the case where $C$ is a linear $(r,\delta)$-LRC. The arguments here serve as a model for the nonlinear extension in Section~\ref{sec:nonlinear}.

For any subset $S \subseteq [n]$, we write $$C(S) \triangleq \{\mb c \in C : \supp(\mb c) \subseteq S\}.$$
This is a subcode of $C$. Since $C \subseteq \F_q^n$ is a linear code, we have two exact sequences
\[
\begin{array}{ c @{\;\to\;} c @{\;\to\;} c @{\;\to\;} c @{\;\to\;} c }
0 & C(S^c)     & C       & \pi_S(C)           & 0, \\[1ex]
0 & C^\perp(S) & C^\perp & \pi_{S^c}(C^\perp) & 0.
\end{array}
\]
As usual, exactness means that the image of each map is precisely the kernel of the next.

\begin{lemma}
	 \label{lem:macw-projection} For any subset $S \subseteq [n]$, we have
\begin{equation*} 
\pi_S(C)^\perp = \pi_S\!\left(C^\perp(S)\right).
\end{equation*}
\end{lemma}

\begin{proof}
This is the standard duality relation between puncturing and shortening; see, for example, \cite[Theorem~1.5.7]{Huffmand}.
\end{proof}

For each $i \in [n]$, let the sets $S_i$ and $T_i=S_i \cup \{i\}$ be given as before.  Denote
\[ C^{(i)} \triangleq \pi_{T_i}(C), \quad k_i \triangleq \dim C^{(i)}, \quad d_i \triangleq d(C^{(i)}) \ge \delta. \]
The code $C^{(i)}$ is linear over $\F_q$ with parameters $[t_i,k_i,d_i]_q$. From the linear Singleton bound applied to $C^{(i)}$,
\begin{align}\label{eq:Singleton-local}
t_i - k_i \,\ge\, d_i - 1 \,\ge\, \delta - 1.
\end{align}
Denote
\begin{align}\label{eq:DEdef}
D^{(i)} \triangleq C^\perp(T_i) = \{\mb c \in C^\perp : \supp(\mb c) \subseteq T_i\}, \qquad E^{(i)} \triangleq \pi_{T_i}(D^{(i)}).
\end{align}
By Lemma~\ref{lem:macw-projection}, $E^{(i)}$ is the dual code of $C^{(i)}$; it has parameters $[t_i, t_i - k_i]_q$. Define 
\begin{align*}
U^{(i)} \triangleq \{\mb c \in C^\perp : \supp(\mb c) \subseteq T_i,\, c_i \neq 0\}.
\end{align*}
Note that $U^{(i)}$ is a subset of $D^{(i)}$, not a subspace.

Recall the following well-known orthogonal array property.

\begin{lemma} \label{lem:OA}
Let $D \subseteq \F_q^t$ be a linear code with dual minimum distance $d^\perp$. Then for every nonempty subset $J \subseteq [t]$ with $|J| \le d^\perp - 1$, the projection $\pi_J: {D} \to \F_q^{J}$ is surjective, and every vector in $\F_q^{J}$ is attained by exactly $|D|/q^{|J|}$ codewords of $D$. 
\end{lemma}

\begin{proof}
This is a standard result: the dual distance being $d^\perp$ implies that any $|J| \le d^\perp - 1$ columns of a generator matrix of $D$ are linearly independent~\cite[Ch.~5]{MacWilliamsd}. Hence $\pi_J: D \to \F_q^{J}$ is a surjective linear map, so each fiber has size exactly $|D|/q^{|J|}$.
\end{proof}

Since $\left(E^{(i)}\right)^\perp=C^{(i)}$ and $d_i=\ud (C^{(i)}) \ge \delta \ge 2$, and the projection $\pi_{T_i} : D^{(i)} \to E^{(i)}$ is a bijective isometry (because every vector of $D^{(i)}$ is supported in $T_i$ by definition \eqref{eq:DEdef}, so projection loses no nonzero coordinates), by Lemma \ref{lem:OA} and \eqref{eq:Singleton-local} we have
\begin{align*}
|U^{(i)}| \,=\, \left(1 - q^{-1}\right)|D^{(i)}| \,=\, \left(1 - q^{-1}\right) q^{t_i - k_i} \ge \left(1 - q^{-1}\right) q^{\delta-1}.
\end{align*}
Denote by $u^{(i)}_\omega$ the number of weight-$\omega$ codewords in $U^{(i)}$, i.e., 
\begin{align} \label{lin:ui} u^{(i)}_\omega \triangleq |\{\mb c \in U^{(i)} : \wt(\mb c) = \omega\}|.
	\end{align} 
The definition gives $u_0^{(i)}=0$ for every $i$.  Also since $u_\omega^{(i)}=0$ for $\omega>t_i$ and $t_i\le R$, we have the zeroth moment identity
\begin{align}\label{eq:zeromom-linear}
\sum_{\omega = 1}^R u^{(i)}_\omega \,=\, |U^{(i)}| \ge \left(1-q^{-1}\right)q^{\delta-1}, \quad \forall \, i \in [n]. 
\end{align}

\begin{remark}  Inequality~\eqref{eq:zeromom-linear} is precisely \cite[Proposition~4.2]{gruicaLRCSDuality2026}.  Note that if $C$ is nondegenerate, then $C^\perp$ has no weight-one codeword, so $u_1^{(i)}=0$ because $D^{(i)}\subseteq C^\perp$ for every $i$; the sum in~\eqref{eq:zeromom-linear} may therefore be started at $\omega=2$, which is the form in which the inequality appears in~\cite{gruicaLRCSDuality2026}.
\end{remark}
\begin{theorem} \label{thm:moment-linear}
Let $C \subseteq \F_q^n$ be a linear $(r,\delta)$-$\lrc$. Then for every $i \in [n]$, and every integer $m$ with $0 \le m \le \delta - 2$, we have 
\begin{equation} \label{eq:moment-linear}
\sum_{\omega = 1}^R \binom{\omega - 1}{m} u^{(i)}_\omega \,=\, \binom{t_i - 1}{m}\,(1 - q^{-1})^m\,|U^{(i)}| .
\end{equation} 
\end{theorem}
\begin{remark}
	The case $m = 0$ of Theorem \ref{thm:moment-linear} corresponds to~\eqref{eq:zeromom-linear}.
\end{remark}
\begin{proof}
We evaluate the left-hand side of~\eqref{eq:moment-linear} in two different ways. First, the term $\binom{\omega - 1}{m} u^{(i)}_\omega$ counts the number of pairs $(\mb c, S)$ such that the codeword $\mb c \in U^{(i)}$ has weight $\omega$ and the subset $S \subseteq \supp(\mb c) \setminus \{i\}$ has cardinality $|S| = m$. Summing over all such $\omega$, the left-hand side of~\eqref{eq:moment-linear} counts the total number of such pairs.

We now count these pairs by choosing the subset $S$ first. There are $\binom{t_i - 1}{m}$ subsets $S \subseteq T_i \setminus \{i\}$ of size $m$. Then for each such fixed $S$, we must count the number of $\mb c \in U^{(i)}$ such that $S \subseteq \supp(\mb c)$, i.e., $c_j \neq 0$ for all $j \in S \cup \{i\}$.

Via the isometry $\pi_{T_i} : D^{(i)} \to E^{(i)}$, this equals the number of $\mb e \in E^{(i)}$ with $e_j \neq 0$ for all $j \in S \cup \{i\}$. Since $|S \cup \{i\}| = m + 1 \le \delta - 1 \le d_i - 1$ where $d_i=d\left((E^{(i)})^\perp\right)$, Lemma~\ref{lem:OA} applies: $\pi_{S \cup \{i\}}$ maps $E^{(i)}$ onto $\F_q^{|S \cup \{i\}|}$ uniformly, with each tuple in $\F_q^{m+1}$ attained exactly $q^{t_i - k_i - (m+1)}$ times. On the other hand, the number of tuples in $\F_q^{m+1}$ with all entries nonzero is $(q-1)^{m+1}$, so
\begin{align*}
|\{\mb e \in E^{(i)} : e_j \neq 0\ \forall j \in S \cup \{i\}\}| &\,=\, (q-1)^{m+1}\, q^{t_i - k_i - (m+1)} \\
&\,=\, (1 - q^{-1})^{m+1} q^{t_i - k_i} \,=\, (1 - q^{-1})^m\, |U^{(i)}|, 
\end{align*}
which gives exactly the right-hand side of~\eqref{eq:moment-linear}.
\end{proof}

\subsection{The general case}\label{sec:nonlinear}

We now extend Theorem~\ref{thm:moment-linear}  to general $(r,\delta)$-LRCs which are not necessarily linear. 

For the code $C \subseteq \F_q^n$, let  $\{A_j\}_{j=0}^n$ be its distance distribution as defined in \eqref{eq:Aj-def}. Let  $\{B_\omega\}_{\omega=0}^n$ be the MacWilliams transform as defined in \eqref{eq:Bw-def}. For each $i \in [n]$ and the recovery set $T_i$, we define the local distance distribution at $i$ as 
\begin{align}\label{eq:Aji-def}
A^{(i)}_j \,\triangleq\, \frac{1}{|C|}\left|\left\{(\mb c, \mb c') \in C^2 : \ud(\pi_{T_i}(\mb c), \pi_{T_i}(\mb c')) = j\right\}\right|, \qquad 0 \le j \le t_i.
\end{align}
Note that unlike the global distribution of a linear code, this local distribution is normalized by $|C|$ rather than by $|\pi_{T_i}(C)|$. Note also that $\sum_j A^{(i)}_j = |C|$.  The collision count at coordinate $i$ is defined as 
\begin{align}\label{eq:lambdai-def}
\lambda^{(i)} \,\triangleq\, A^{(i)}_0 \,=\, \frac{1}{|C|}\left|\{(\mb c, \mb c') \in C^2 : \pi_{T_i}(\mb c) = \pi_{T_i}(\mb c')\}\right|.
\end{align}
By definition, the MacWilliams transform of $\{A^{(i)}_j\}_{j=0}^{t_i}$ on length $t_i$ is
\begin{align}\label{eq:Bwi-def}
B^{(i)}_\omega \,\triangleq\, \frac{1}{|C|}\sum_{j = 0}^{t_i} A^{(i)}_j K_\omega(j; t_i, q), \qquad 0 \le \omega \le t_i.
\end{align}
Using the basic summation identity \eqref{eq:krawtchouk-sum},
we obtain 
$$B^{(i)}_0 = 1, \quad \sum_{\omega=0}^{t_i} B^{(i)}_{\omega} = \frac{q^{t_i}}{|C|} \lambda^{(i)}.$$ 
We also have the character-sum representation of $B^{(i)}_\omega$ as (see \cite[Ch.~5,~Theorem~6]{MacWilliamsd})
\begin{align}\label{eq:Bwi-charsum}
B^{(i)}_\omega \,=\, \frac{1}{|C|^2}\sum_{\substack{\mb v \in \F_q^{T_i}\\ \wt(\mb v) = \omega}}\left|\sum_{\mb c \in C}\chi(\mb v \cdot \pi_{T_i}(\mb c))\right|^2 \,\ge\, 0, \quad \forall \, 0 \le \omega \le t_i.
\end{align}
The condition $\ud(\pi_{T_i}(C)) \ge \delta$ implies that 
\begin{align}\label{eq:Aivanish}
A^{(i)}_j \,=\, 0, \qquad \forall \, 1 \le j \le \delta - 1.
\end{align}
We decompose $B^{(i)}_\omega$ in~\eqref{eq:Bwi-charsum} as \begin{align*}
B^{(i)}_\omega \,=\, u^{(i)}_\omega + r^{(i)}_\omega, 
\end{align*}
where
\begin{align}\label{eq:uidef-nonlin}
u^{(i)}_\omega &\,\triangleq\, \frac{1}{|C|^2}\sum_{\substack{\mb v \in \F_q^{T_i}\\ \wt(\mb v) = \omega\\ v_i \neq 0}}\left|\sum_{\mb c \in C}\chi(\mb v \cdot \pi_{T_i}(\mb c))\right|^2 \ge 0, \\
\label{eq:uidef-nonlin-r} 
r^{(i)}_\omega &\,\triangleq\, \frac{1}{|C|^2}\sum_{\substack{\mb v \in \F_q^{T_i}\\ \wt(\mb v) = \omega\\ v_i = 0}}\left|\sum_{\mb c \in C}\chi(\mb v \cdot \pi_{T_i}(\mb c))\right|^2 \ge 0. 
\end{align}
The definition gives $u_0^{(i)}=0$ for every $i \in [n]$. 

Define $A^{(S_i)}_j$ and $B^{(S_i)}_\omega$ exactly as in~\eqref{eq:Aji-def}--\eqref{eq:Bwi-def} but with the set $T_i$ replaced by $S_i$. Note that $B^{(S_i)}_\omega$ is of length $|S_i|=t_i-1$. In the character-sum representation of $r_{\omega}^{(i)}$ in~\eqref{eq:uidef-nonlin-r}, for any $\mb{v} \in \F_q^{T_i}$ with $\wt(\mb{v})=\omega$ and $v_i =0$, let $\mb v' \in \F_q^{S_i}$ be the restriction of this $\mb{v}$ to $S_i$, then $\mb v \cdot \pi_{T_i}(\mb c) = \mb v' \cdot \pi_{S_i}(\mb c)$, so we have 
\begin{align*}
r^{(i)}_\omega \,=\, B^{(S_i)}_\omega.
\end{align*}
We also have the following result. 
\begin{lemma} \label{lem:recovery}
Let $C \subseteq \F_q^n$ be an $(r,\delta)$-$\lrc$ with fixed recovery sets $S_i$ and $T_i=S_i \cup \{i\}$ for each $i$. Then for all $\mb c, \mb c' \in C$, $\pi_{S_i}(\mb c) = \pi_{S_i}(\mb c')$ implies that $\pi_{T_i}(\mb c) = \pi_{T_i}(\mb c')$. Consequently,
\[ A^{(S_i)}_0 \,=\, A^{(i)}_0 \,=\, \lambda^{(i)}, \qquad A^{(S_i)}_j \,=\, 0 \ \mbox{for}\ 1 \le j \le \delta - 2. \]
\end{lemma}

\begin{proof}
If $\pi_{S_i}(\mb c) = \pi_{S_i}(\mb c')$ but $c_i \neq c'_i$, then $\ud(\pi_{T_i}(\mb c), \pi_{T_i}(\mb c')) = 1$, contradicting $\ud(\pi_{T_i}(C)) \ge \delta \ge 2$. The first claim follows, so does the equality $A^{(S_i)}_0 = A^{(i)}_0$. Next, if $\ud(\pi_{S_i}(\mb c), \pi_{S_i}(\mb c'))= j$ for some $j$ where $1 \le j \le \delta - 2$, then $\ud(\pi_{T_i}(\mb c), \pi_{T_i}(\mb c')) \in \{j, j+1\} \subseteq \{1, \ldots, \delta - 1\}$, contradicting $\ud(\pi_{T_i}(C)) \ge \delta$.
\end{proof}

\begin{theorem} \label{thm:moment-nonlinear}
Let $C \subseteq \F_q^n$ be the $(r,\delta)$-$\lrc$ described above. Then for every $i \in [n]$ and every integer $m$ with $0 \le m \le \delta - 2$, we have 
\begin{align}\label{eq:moment-nonlinear}
\sum_{\omega = 1}^R \binom{\omega - 1}{m}\, u^{(i)}_\omega \,=\, \binom{t_i - 1}{m}\,(1 - q^{-1})^{m+1}\, \frac{q^{t_i}}{|C|}\, \lambda^{(i)}.
\end{align}
\end{theorem}

\begin{proof}
By the definition of $B^{(i)}_\omega$ in \eqref{eq:Bwi-def}, we have
\[
\sum_\omega \binom{\omega}{m} B^{(i)}_\omega \,=\, \frac{1}{|C|}\sum_j A^{(i)}_j \sum_\omega \binom{\omega}{m} K_\omega(j; t_i, q).
\]
Applying Lemma~\ref{lem:krawtchouk-moment} with length $t_i$, and using the assumption $m\leq \delta-2$,  we obtain 
\[\sum_\omega \binom{\omega}{m} K_\omega(j; t_i, q)=0 \textrm{ \ for \ } j\geq \delta.\]
Indeed, $j\ge\delta$ implies $j>m$, so Lemma~\ref{lem:krawtchouk-moment} gives zero, while the local minimum-distance condition gives $A^{(i)}_j=0$ for $1\leq j\leq \delta-1$ by~\eqref{eq:Aivanish}.  Hence only the term $j=0$ contributes to the above sum. Therefore,
\begin{align}\label{eq:Tmi}
\sum_\omega \binom{\omega}{m} B^{(i)}_\omega \,=\, \frac{1}{|C|}A^{(i)}_0\, {q^{t_i - m}}(q - 1)^m\binom{t_i}{m} \,=\, \lambda^{(i)} \frac{q^{t_i}}{|C|}(1 - q^{-1})^m \binom{t_i}{m}.
\end{align}

For the punctured set $S_i$, by Lemma~\ref{lem:recovery} we have $A^{(S_i)}_j = 0$ for $1 \le j \le \delta - 2$ and $A^{(S_i)}_0 = \lambda^{(i)}$. Lemma~\ref{lem:krawtchouk-moment} applied with length $t_i - 1$ gives, for $0 \le m \le \delta - 2$,
\begin{align}\label{eq:Tms}
\sum_\omega \binom{\omega}{m} B^{(S_i)}_\omega \,=\, \frac{1}{|C|}\lambda^{(i)} {q^{t_i - 1 - m}}(q - 1)^m\binom{t_i - 1}{m} \,=\, \lambda^{(i)} \frac{q^{t_i - 1}}{|C|}(1 - q^{-1})^m \binom{t_i - 1}{m}.
\end{align}
Subtracting~\eqref{eq:Tms} from~\eqref{eq:Tmi} and using $r^{(i)}_\omega = B^{(S_i)}_\omega$, we obtain 
\[
\sum_\omega \binom{\omega}{m} u^{(i)}_\omega \,=\, \lambda^{(i)} \frac{q^{t_i}}{|C|}(1 - q^{-1})^m\!\left[\binom{t_i}{m} - q^{-1}\binom{t_i - 1}{m}\right].
\]
Using Pascal's identity $\binom{t_i}{m} = \binom{t_i - 1}{m} + \binom{t_i - 1}{m - 1}$, the right-hand side can be rewritten as 
\begin{align}\label{eq:Tm}
\tau_m \,\triangleq\, \sum_\omega \binom{\omega}{m} u^{(i)}_\omega \,=\, \lambda^{(i)} \frac{q^{t_i}}{|C|}(1 - q^{-1})^m\!\left[(1 - q^{-1})\binom{t_i - 1}{m} + \binom{t_i - 1}{m - 1}\right].
\end{align}
Set
\[
\sigma_m \triangleq \sum_\omega \binom{\omega-1}{m}u^{(i)}_\omega,
\qquad \sigma_{-1}=0.
\]
The identity
\[
\binom{\omega}{m}=\binom{\omega-1}{m}+\binom{\omega-1}{m-1}
\]
gives $\tau_m=\sigma_m+\sigma_{m-1}$.  Thus $\sigma_0=\tau_0=(1-q^{-1})q^{t_i}\lambda^{(i)}/|C|$.  For $m\ge1$, assuming the formula for $\sigma_{m-1}$ and using~\eqref{eq:Tm}, we obtain
\begin{align*}
\sigma_m
&=\tau_m-\sigma_{m-1}\\
&=\lambda^{(i)}\frac{q^{t_i}}{|C|}(1-q^{-1})^m
\left[(1-q^{-1})\binom{t_i-1}{m}+\binom{t_i-1}{m-1}\right]
-\lambda^{(i)}\frac{q^{t_i}}{|C|}(1-q^{-1})^m\binom{t_i-1}{m-1}\\
&=\lambda^{(i)}\frac{q^{t_i}}{|C|}(1-q^{-1})^{m+1}\binom{t_i-1}{m}.
\end{align*}
This is exactly~\eqref{eq:moment-nonlinear}. 
\end{proof}

\begin{remark}\label{rmk:linear-corollary}
When $C$ is a linear $[n,k]_q$ code, then $\lambda^{(i)} = q^{k - k_i}$, so $\frac{q^{t_i}}{|C|}\lambda^{(i)} =  q^{t_i - k_i}$. Via Lemma~\ref{lem:macw-projection}, the character-sum quantity $u^{(i)}_\omega$ of~\eqref{eq:uidef-nonlin} equals exactly the number of weight-$\omega$ codewords in $U^{(i)}$ (the local set of Section~\ref{sec:linear}). Substituting this into~\eqref{eq:moment-nonlinear} gives
\[
\sum_{\omega =1}^R \binom{\omega - 1}{m}\, {u_{\omega}^{(i)}} \,=\, \binom{t_i - 1}{m}(1 - q^{-1})^{m+1} q^{t_i - k_i} \,=\, \binom{t_i - 1}{m}(1 - q^{-1})^m |U^{(i)}|,
\]
which is exactly Theorem~\ref{thm:moment-linear}, recovered as a special case. 
\end{remark}

\section{The LP Bounds}\label{sec:main}
We now convert Theorem \ref{thm:moment-nonlinear}, the local moment identities of the previous section, into global linear constraints on the distance distribution of $C$.  The auxiliary quantities $u^{(i)}_\omega$ measure the part of the local MacWilliams transform that involves the recovered coordinate $i$.  The idea is to average over all coordinates.

\subsection{The zeroth moment constraint}\label{sec:global}

Define
\begin{align*}
	u_\omega \,\triangleq\, \frac{1}{n}\sum_{i = 1}^n u^{(i)}_\omega, \quad \forall \,1 \le \omega \le R.
\end{align*}
We now derive two linear constraints linking $u_\omega$ to the distance distribution $\{A_j\}$ of $C$.

The $m = 0$ case of Theorem~\ref{thm:moment-nonlinear} gives
\begin{align}\label{eq:sumu}
	\sum_{\omega = 1}^R u_\omega \,=\, (1 - q^{-1}) \cdot \frac{1}{n}\sum_{i = 1}^n \frac{q^{t_i}}{|C|}\, \lambda^{(i)}.
\end{align}
\begin{lemma}\label{lem:zeromoment}
	Let $C$ be the $(r,\delta)$-$\lrc$ as described before. Then 
	\begin{align*}
		\sum_{\omega=1}^R u_\omega \,\ge\, (1 - q^{-1})\, q^{\delta - 1}.
	\end{align*}
\end{lemma}

\begin{proof}
	Fix $i \in [n]$ and let
	$$V = \pi_{T_i}(C), \quad N_{\mb{v}} = \left|\left\{\mb{c} \in C: \pi_{T_i}(\mb c)=\mb{v}\right\}\right| \quad \forall \mb{v} \in V.$$
	Since $V$ is a code of length $t_i$ with minimum distance $d(V) \ge \delta \ge 2$, the Singleton bound gives 
	$$|V|=|\pi_{T_i}(C)| \le q^{t_i-\delta+1}.$$
	By the definition of $\lambda^{(i)}$ in~\eqref{eq:lambdai-def}, we can write   
	$$\lambda^{(i)}=\frac{1}{|C|} \sum_{\mb{v} \in V}N_{\mb{v}}^2. $$
	Now noting the identity $\sum_{\mb{v} \in V} N_{\mb{v}}=|C|$ and applying the Cauchy-Schwarz inequality, we obtain $$\lambda^{(i)} \ge  \frac{|C|}{|\pi_{T_i}(C)|} \ge \frac{|C|}{q^{t_i-\delta+1}} \quad \Longrightarrow  \quad \frac{q^{t_i}}{|C|} \lambda^{(i)} \ge q^{\delta - 1}.$$ 
	Substituting into~\eqref{eq:sumu}, we obtain 
	\[
	\sum_{\omega=1}^R u_\omega \,=\, (1 - q^{-1})\cdot \frac{1}{n}\sum_{i=1}^n \frac{q^{t_i}}{|C|}\lambda^{(i)} \,\ge\, (1 - q^{-1})\, q^{\delta - 1}. \qedhere
	\]
\end{proof}

\begin{remark}
When $C$ is linear with $\dim C=k$, we have $\lambda^{(i)} = q^{k - k_i}$, and~\eqref{eq:sumu} reduces to
\[
\sum_{\omega = 1}^R {u_\omega} \,=\, (1 - q^{-1}) \cdot \frac{1}{n}\sum_i q^{t_i - k_i}, 
\]
and the bound $\frac{q^{t_i}}{|C|} \, \lambda^{(i)} \ge q^{\delta - 1}$ reduces to $q^{t_i - k_i} \ge q^{\delta - 1}$, which is exactly the linear Singleton bound~\eqref{eq:Singleton-local} restated. 
\end{remark}

\begin{lemma}\label{lem:coupling}
	\begin{align}\label{eq:coupling}
		B_{\omega} \,\ge\, \frac{n}{\omega}\, u_\omega, \quad \forall \, 1 \le \omega \le R.
	\end{align}
\end{lemma}

\begin{proof}
	We may identify any $\mb v \in \F_q^{T_i}$ with its extension $\mb{v}'$ in $\F_q^n$ such that $\supp(\mb{v}') \subseteq T_i$, so that $\mb v \cdot \pi_{T_i}(\mb c) = \mb v' \cdot \mb c$, that is, 
	\[
	u^{(i)}_\omega \,=\, \frac{1}{|C|^2}\sum_{\substack{\mb v \in \F_q^{T_i}\\ 
			\wt(\mb v) = \omega\\
		v_i \ne 0}}\!\left|\sum_{\mb c \in C}\chi(\mb v \cdot \pi_{T_i}(\mb c))\right|^2=\, \frac{1}{|C|^2}\sum_{\substack{\mb v \in \F_q^n\\ 
			\supp(\mb{v}) \subseteq T_i\\
			\wt(\mb v) = \omega\\
		v_i \ne 0}}\!\left|\sum_{\mb c \in C}\chi(\mb v \cdot \mb c)\right|^2.
	\]
	Then summing $u^{(i)}_\omega$ over $i$ we have 
	\[
	\sum_{i=1}^n u^{(i)}_\omega \,=\, \frac{1}{|C|^2}\sum_{\substack{\mb v \in \F_q^n\\ \wt(\mb v) = \omega}}\!\left|\sum_{\mb c \in C}\chi(\mb v \cdot \mb c)\right|^2 \cdot \left|\left\{i \in [n] : \supp(\mb v) \subseteq T_i,\, v_i \neq 0\right\}\right|.
	\]
	For each fixed $\mb v$ of weight $\omega$, the count $|\{i : \supp(\mb v) \subseteq T_i,\, v_i \neq 0\}|$ is at most $|\{i: v_i \neq 0\}| = \omega$. Therefore
	\[
	\sum_{i=1}^n u^{(i)}_\omega \,\le\, \omega \cdot \frac{1}{|C|^2}\sum_{\substack{\mb v \in \F_q^n\\ \wt(\mb v) = \omega}}\!\left|\sum_{\mb c \in C}\chi(\mb v \cdot \mb c)\right|^2 \,=\, \omega\cdot B_\omega.
	\]
	Dividing by $n$ on both sides gives the desired result~\eqref{eq:coupling}.
\end{proof}
\begin{remark}
In the linear case, Lemma~\ref{lem:coupling} reads as $A'_\omega \ge \frac{n}{\omega}\,u_\omega$, where $A'_\omega=W_\omega(C^\perp)=B_\omega$.  One can prove this directly from the fact that $u^{(i)}_\omega$ counts the weight-$\omega$ codewords in $U^{(i)}$ (see~\eqref{lin:ui}).  \end{remark}

\subsection{Box and convex hull constraints}\label{sec:box-CH-constraints}
The LP developed in~\cite{gruicaLRCSDuality2026} essentially uses only the zeroth-moment information captured by Lemmas~\ref{lem:zeromoment} and~\ref{lem:coupling}.  This captures the full locality information when $\delta=2$; see \cite[Lemma~4.4]{gruicaLRCSDuality2026}.  For $\delta\ge3$, the higher local moments yield additional valid inequalities.

Assume that $\delta \ge 3$. We derive two families of linear constraints on $\{u_\omega\}$ from Theorem~\ref{thm:moment-nonlinear}: the box constraints (Section~\ref{ssec:box}) and the stronger convex hull constraints (Section~\ref{ssec:CH}).

\subsubsection{The box constraints}\label{ssec:box}

Averaging Theorem~\ref{thm:moment-nonlinear} over $i \in [n]$, we obtain 
\[
\sum_{\omega=1}^R \binom{\omega - 1}{m} u_\omega \,=\, (1 - q^{-1})^{m+1} \cdot \frac{1}{n}\sum_{i=1}^n \binom{t_i - 1}{m} \frac{q^{t_i}}{|C|}\, \lambda^{(i)}, \qquad \forall \, 0 \le m \le \delta - 2.
\]
Dividing both sides by~\eqref{eq:sumu}, the normalized $m$-th factorial moment is
\[
\frac{\sum_\omega \binom{\omega - 1}{m} u_\omega}{\sum_\omega u_\omega} \,=\, (1 - q^{-1})^m \cdot \frac{\sum_i \binom{t_i - 1}{m}\, q^{t_i}\, \lambda^{(i)}}{\sum_i q^{t_i}\,\lambda^{(i)}}.
\]
Since $t \mapsto \binom{t-1}{m}$ is nondecreasing in $t$, and strictly increasing for $t \ge m+1$, and since $\delta \le t_i \le R$ with $m \le \delta-2$ for each $i$, we obtain
\begin{align}\label{eq:box-constraint}
\lambda(m)\sum_{\omega=1}^R u_\omega \,\le\, \sum_{\omega=1}^R \binom{\omega - 1}{m} u_\omega \,\le\, \mu(m) \sum_{\omega=1}^R u_\omega, \qquad \forall \, 1 \le m \le \delta - 2,
\end{align}
where
\begin{align*}
\lambda(m) \,\triangleq\, \binom{\delta - 1}{m}(1 - q^{-1})^m, \qquad \mu(m) \,\triangleq\, \binom{R - 1}{m}(1 - q^{-1})^m.
\end{align*}
We refer to~\eqref{eq:box-constraint} as the \emph{box  constraints}.

\subsubsection{The convex hull constraints}\label{ssec:CH}

The box constraints bound each factorial moment separately.  We now exploit the joint structure of all moments by requiring them to arise from one common mixture of admissible local lengths.

For each integer $t \in \{\delta, \delta + 1, \ldots, R\}$, define the vector
\begin{align*}
\vec P(t) \,\triangleq\, \left[\binom{t-1}{1}(1 - q^{-1}),\ \binom{t-1}{2}(1 - q^{-1})^2,\ \ldots,\ \binom{t-1}{\delta - 2}(1 - q^{-1})^{\delta - 2}\right] \,\in\, \mathbb R^{\delta - 2},
\end{align*}
and let
\begin{align*}
\mathcal V \,\triangleq\, \left\{\vec P(\delta),\, \vec P(\delta + 1),\, \ldots,\, \vec P(R)\right\}, \qquad \mathcal H \,\triangleq\, \mathrm{Conv}(\mathcal V) \subseteq \mathbb R^{\delta - 2}.
\end{align*}
The hull $\mathcal H$ may not be of full dimension in $\mathbb R^{\delta-2}$. Fix a complete \(H\)-representation of \(\mathcal H\),
\begin{align}
\label{eq:H-representation}
\mathcal H
=
\left\{
\mathbf y\in\mathbb R^{\delta-2}:
F\mathbf y=\mathbf f,\;
G\mathbf y\le\mathbf g
\right\}.
\end{align}
Here \(F\mathbf y=\mathbf f\) describes the affine hull of
\(\mathcal H\), while \(G\mathbf y\le\mathbf g\) describes its
relative facet inequalities.  If \(\mathcal H\) is
full-dimensional in \(\mathbb R^{\delta-2}\), then the affine-hull
system is empty.

\begin{theorem} \label{thm:convex-hull}
Let $C$ be the $(r,\delta)$-$\lrc$ as described before and assume that $\delta \ge 3$. Define
\[
\vec M \,\triangleq\, \left[\sum_{\omega=1}^R \binom{\omega - 1}{1} u_\omega,\sum_{\omega=1}^R \binom{\omega - 1}{2} u_\omega,\dots,\sum_{\omega=1}^R \binom{\omega - 1}{\delta-2} u_\omega\right] \in \mathbb R^{\delta-2}. 
\]
Then the normalized moment vector $\vec M/\sum_{\omega=1}^R u_\omega$ lies in $\mathcal H$, or  equivalently,
\begin{align}\label{eq:CH-master}
\vec M \,\in\, \left(\sum_{\omega=1}^R u_\omega\right) \cdot \mathcal H.
\end{align}
Equivalently, under the complete \(H\)-representation
\eqref{eq:H-representation}, the moment vector satisfies
\begin{align}
\label{eq:CH-Hrep}
F\vec M
=
\mathbf f\sum_{\omega=1}^R u_\omega, \qquad 
G\vec M
\le
\mathbf g\sum_{\omega=1}^R u_\omega.
\end{align}
\end{theorem}

\begin{proof}
By Theorem~\ref{thm:moment-nonlinear}, 
\begin{align*}
    \sum_{\omega} \binom{\omega-1}{m} u_\omega^{(i)}& =\binom{t_i-1}{m}\left(1-q^{-1}\right)^{m+1}\frac{q^{t_i}}{|C|}\lambda^{(i)}\\
    & = (1-q^{-1})\frac{q^{t_i}\lambda^{(i)}}{|C|}\left(\binom{t_i-1}{m} \left(1-q^{-1}\right)^{m}\right).
\end{align*}
Averaging over $i$, we have 
\[
\vec M \,=\, (1 - q^{-1})\cdot \frac{1}{n}\sum_{i=1}^n \frac{q^{t_i}\lambda^{(i)}}{|C|} \vec P(t_i) \,\overset{\eqref{eq:sumu}}{=}\, \left(\sum_\omega u_\omega\right) \cdot \frac{\sum_i q^{t_i} \lambda^{(i)} \vec P(t_i)}{\sum_j q^{t_j}\lambda^{(j)}}.
\]
The fraction on the right is a convex combination of points $\vec P(t_i)$ in $\mathcal V$, with nonnegative weights $q^{t_i}\lambda^{(i)} / \sum_j q^{t_j}\lambda^{(j)}$ summing to 1, hence lies in $\mathcal H$.

The relations in~\eqref{eq:CH-Hrep} now follow from the
\(H\)-representation~\eqref{eq:H-representation}.
\end{proof}

\begin{remark}
The set \(\mathcal H\) has at most \(R-\delta+1\) vertices.
Its complete \(H\)-representation consists of the equations
defining its affine hull together with its relative facet
inequalities.  These depend only on \((q,\delta,R)\), and can be
computed by a standard polyhedral package.  When \(\mathcal H\)
is full-dimensional in \(\mathbb R^{\delta-2}\), the affine-hull
equations are absent.
\end{remark}

\subsection{The linear program}\label{sec:LP}

We now assemble the constraints derived in the previous sections into a single linear program that bounds $|C|$ for any $(r, \delta)$-LRC, linear or not.

Let $C \subseteq \F_q^n$ be an $(r,\delta)$-$\lrc$ of minimum distance $d>0$. Let $\{A_j\}_{j=0}^n$ be the distance distribution of $C$ as defined in \eqref{eq:Aj-def}. Then we have
\[A_0=1, A_j = 0 \quad \forall \, 1 \le j \le d-1, \mbox{ and } A_j \ge 0 \quad \forall \, 0 \le j \le n.\]
The objective is to maximize the code size $|C|=\sum_{j=0}^n A_j$. 

We define new variables 
\begin{equation*}  
	v_\omega \triangleq |C| u_\omega, \quad \forall \, 1 \le \omega \le R.
\end{equation*}
For $\omega=1$, by using \eqref{eq:uidef-nonlin} and \eqref{eq:Bw-def}, we see that 
\begin{align}\label{eq:u1}
	u^{(i)}_1 &\,\triangleq\, \frac{1}{|C|^2}\sum_{a \in \F_q^*}\left|\sum_{\mb c \in C}\chi(a c_i)\right|^2, \qquad B_1=\sum_{i=1}^n u_1^{(i)}=nu_1, 
\end{align}
and hence 
\begin{eqnarray} \label{lp:1-moment1} \sum_{j=0}^nA_jK_1(j;n,q)=nv_1, \quad B_1=\frac{n}{|C|} v_1.\end{eqnarray}
As for $2 \le \omega \le R$, substituting the MacWilliams identity \eqref{eq:Bw-def} into  Lemma \ref{lem:coupling} yields 
\begin{equation*}
	B_\omega \,|C| =\sum_{j=0}^n A_j K_\omega(j; n, q)\ge \frac{n}{\omega} u_\omega |C| =\frac{n}{\omega} v_\omega, \quad \forall \, 2 \le \omega \le R.
\end{equation*}
Lemma \ref{lem:zeromoment} becomes 
\begin{equation*} 
	\sum_{\omega=1}^R v_\omega \ge (1 - q^{-1}) q^{\delta-1} |C|. 
\end{equation*}
Similarly in \eqref{eq:box-constraint} and Theorem \ref{thm:convex-hull} we use the new variables $v_\omega$. We obtain the following linear program.  Given positive integers $n,q,d,r,\delta$ where $n \ge d \ge \delta \ge 2$, let $R\triangleq \min\{n,r+\delta-1\}$. \\ 
\textbf{Maximize} $\displaystyle \sum_{j=0}^n A_j \,\,$ \textbf{subject to:}
\begin{align}
	& A_0 = 1, \quad A_j = 0 \;\;\forall\, 1 \le j \le d - 1, \quad A_j \ge 0 \;\;\forall\, 0 \le j \le n, \label{lp:primal_dist} \\
	& \sum_{j=0}^n A_j K_\omega(j; n, q) \ge 0, \quad \forall\, 0 \le \omega \le n, \label{lp:dual_dist} \\
	& \, \, v_\omega \ge 0, \quad \forall\, 1 \le \omega \le R, \notag \\
	& \sum_{\omega=1}^R v_\omega \ge (1 - q^{-1})q^{\delta-1} \sum_{j=0}^n A_j, \label{lp:zeroth_moment} \\
	&\sum_{j=0}^nA_jK_1(j;n,q)=nv_1, \label{lp:1-moment}\\
	& \sum_{j=0}^n A_j K_\omega(j; n, q) \ge \frac{n}{\omega} v_\omega, \quad \forall\, 2 \le \omega \le R, \label{lp:coupling}
	\end{align}
\textbf{Moment constraint} (one of the following two versions).

For convenience, put
\begin{align*}
s_v
&\triangleq
\sum_{\omega=1}^R v_\omega,
\\
\vec M_v
&\triangleq
\left(
\sum_{\omega=1}^R
\binom{\omega-1}{m}v_\omega
\right)_{m=1}^{\delta-2}.
\end{align*}

For the box LP, impose
\begin{align}
\label{lp:box}
\lambda(m)s_v
\le
\sum_{\omega=1}^R
\binom{\omega-1}{m}v_\omega
\le
\mu(m)s_v,
\qquad
1\le m\le\delta-2.
\end{align}

For the convex-hull LP, impose
\begin{equation}
\label{lp:convexhull}
\begin{aligned}
F\vec M_v = \mathbf f\,s_v,\quad G\vec M_v \le \mathbf g\,s_v.
\end{aligned}
\end{equation}
We call the formulation with~\eqref{lp:box} the
\emph{box LP}, and the formulation with~\eqref{lp:convexhull}
the \emph{convex-hull LP}.  Both are finite linear programs
with \(n+R+1\) variables, or \(n+R\) free variables after
eliminating \(A_0=1\).  In our computations, a complete
\(H\)-representation of \(\mathcal H\), including both its
affine-hull equations and its relative facet inequalities, is
computed in \texttt{SageMath}; the resulting LP is assembled
and solved using \texttt{PuLP}. We summarize our results below. 

\begin{theorem} \label{thm:lp_bound}
	Let $C \subseteq \mathbb{F}_q^n$ be an $(r, \delta)$-LRC with minimum distance at least $d$. Then $|C|$ is at most the optimal value of the box LP, and also at most the optimal value of the convex-hull LP.
\end{theorem}
\begin{proof}
Let $C$ be any $(r,\delta)$-LRC with minimum distance at least $d$.
Its distance distribution $\{A_j\}$ satisfies the standard Delsarte constraints \eqref{lp:primal_dist}--\eqref{lp:dual_dist}.  Defining $v_\omega=|C|u_\omega$, Lemmas~\ref{lem:zeromoment} and~\ref{lem:coupling}
give \eqref{lp:zeroth_moment} and \eqref{lp:coupling}, \eqref{lp:1-moment} follows from \eqref{lp:1-moment1}.  The box constraints follow from
\eqref{eq:box-constraint}, while
\eqref{eq:CH-master} and the complete
\(H\)-representation~\eqref{eq:H-representation} give
\eqref{lp:convexhull}. Hence the variables associated with $C$ give a feasible point of the corresponding LP with objective value $|C|$. Therefore the LP optimum is at least $|C|$, which proves the claimed upper bound.
\end{proof}
\begin{remark} \label{7:balanced} We call $C$ \emph{balanced} if $v_1=0$. Since 
\[u_1=\frac{1}{n}\sum_i u_1^{(i)} \quad \text{and}\quad v_1=|C|u_1,\] 
and all $u_1^{(i)}$ are nonnegative, this is equivalent to $u_1^{(i)}=0$ for every $i\in[n]$ (see \eqref{eq:u1}). This condition has the following interpretation.  For $a\in\F_q$, let
\[
N_{i,a}=|\{\mb c\in C:c_i=a\}|.
\]
It is easy to see that 
\[
u_1^{(i)}=0
\quad\Longleftrightarrow\quad
N_{i,a}=|C|/q\quad\text{for every }a\in\F_q.
\]
Thus $C$ is balanced precisely when every symbol of $\F_q$ occurs exactly $|C|/q$ times in each coordinate.  If $C$ is linear, then $C$ is balanced if and only if $C$ is nondegenerate.
\end{remark}

\begin{remark} \label{rmk:CH-vs-box}
	Projecting $\mathcal H$ onto the $m$-th coordinate axis gives the interval $[\lambda(m), \mu(m)]$. Therefore the convex hull constraint~\eqref{eq:CH-master} implies the box constraint~\eqref{eq:box-constraint}, but is generally strictly stronger. In particular, 
	\begin{enumerate}[noitemsep, topsep=2pt]
		\item[(1)] for $\delta = 2$, the box constraint~\eqref{eq:box-constraint} and  the convex hull constraint~\eqref{eq:CH-master} are both  vacuous; 
		\item[(2)] for $\delta = 3$, the moment space is one-dimensional, so the convex hull is an interval coinciding with the box at $m = 1$; 
		\item[(3)] for $\delta \ge 4$ and $R \ge \delta+1$, the convex hull is a strict refinement of the box; when $R = \delta$ the set $\mathcal V$ is a single point and the two coincide.
	\end{enumerate}
\end{remark}

\begin{example}\label{ex:delta4}
Here for $\delta=4$, we give a simple example in which the convex-hull geometry is explicit. Take $q=2$ and suppose that the possible local lengths are
\[
        t\in\{4,5,6,7,8,9\};
\]
for instance this occurs when $R=9$.  Since $1-q^{-1}=1/2$, the moment points are
\[
\begin{array}{c|cccccc}
 t      &4&5&6&7&8&9\\ \midrule
\vec P(t) = \begin{bmatrix} (t-1)/2 \\ \binom{t-1}{2}/4\end{bmatrix}      & \begin{bmatrix}1.5\\0.75\end{bmatrix} &
\begin{bmatrix}2\\1.5\end{bmatrix} &
\begin{bmatrix}2.5\\2.5\end{bmatrix} &
\begin{bmatrix}3\\3.75\end{bmatrix} &
\begin{bmatrix}3.5\\5.25\end{bmatrix} &
\begin{bmatrix}4\\7\end{bmatrix} 
\end{array}
\]
and the feasible normalized vector must lie in
\[
\mathcal H=\operatorname{Conv}\left\{\begin{bmatrix}1.5\\0.75\end{bmatrix}, \begin{bmatrix}2\\1.5\end{bmatrix}, 
\begin{bmatrix}2.5\\2.5\end{bmatrix}, 
\begin{bmatrix}3\\3.75\end{bmatrix}, 
\begin{bmatrix}3.5\\5.25\end{bmatrix}, 
\begin{bmatrix}4\\7\end{bmatrix} \right\} \subseteq \mathbb R^2.
\]
The box relaxation remembers only the coordinate-wise ranges
\[
1.5\le y_1\le4,
\qquad
0.75\le y_2\le7.
\]
It therefore admits points outside $\mathcal H$.  For example, the point $\begin{bmatrix}2.8\\ 6\end{bmatrix}$ satisfies
\[
1.5\le 2.8\le 4,\qquad 0.75\le 6\le 7,
\]
and hence satisfies all box inequalities.  However, it lies outside the convex hull $\mathcal H$. Indeed, $\mathcal H$ has the facet inequality
\[
        y_2\le \frac52 y_1-3,
\]
so at $y_1=2.8$, any point in $\mathcal H$ must have $y_2\le 4$, whereas the chosen point has $y_2=6$.  This is exactly the information that the convex-hull LP keeps and the box LP discards. 

Figure~\ref{fig:delta4-hull} illustrates the situation.  The dashed rectangle is the box relaxation, while the shaded red polygon is the actual convex hull $\mathcal H$.  The marked orange point satisfies all box inequalities but lies outside the hull. 

\begin{figure}[H]
\centering
\begin{tikzpicture}[x=2.3cm,y=0.85cm]
    \draw[help lines, gray!25, step=0.5] (1.25,0.5) grid (4.3,7.35);

    \fill[blue!8] (1.5,0.75) rectangle (4,7);
    \draw[blue!70!black, thick, dashed] (1.5,0.75) rectangle (4,7);

    \fill[red!22] (1.5,0.75) -- (2,1.5) -- (2.5,2.5) -- (3,3.75) -- (3.5,5.25) -- (4,7) -- cycle;
    \draw[red!80!black, very thick] (1.5,0.75) -- (2,1.5) -- (2.5,2.5) -- (3,3.75) -- (3.5,5.25) -- (4,7) -- cycle;
    \draw[purple!90!black, line width=1.1pt] (1.5,0.75) -- (4,7);

    \foreach \x/\y/\t in {1.5/0.75/4,2/1.5/5,2.5/2.5/6,3/3.75/7,3.5/5.25/8,4/7/9}{
        \filldraw[black] (\x,\y) circle (1.5pt);
        \node[black, above left=1pt] at (\x,\y) {\scriptsize $P(\t)$};
    }

    \filldraw[orange!90!black] (2.8,6) circle (1.5pt);

    \draw[->, thick] (1.25,0.5) -- (4.35,0.5) node[right] {$y_1$};
    \draw[->, thick] (1.25,0.5) -- (1.25,7.45) node[above] {$y_2$};
    \foreach \v in {1.5,2,2.5,3,3.5,4}{
        \draw (\v,0.47) -- (\v,0.53);
        \node[below] at (\v,0.45) {\scriptsize $\v$};
    }
    \foreach \v in {1,2,3,4,5,6,7}{
        \draw (1.22,\v) -- (1.28,\v);
        \node[left] at (1.20,\v) {\scriptsize $\v$};
    }
    \node[blue!60!black]  at (3.65,1.45) {\small box};
    \node[red!70!black]   at (2.72,3.42) {\small $\mathcal H$};
    \node[purple!90!black, align=center] at (2,5.2)
    {\scriptsize $y_2\le \frac52 y_1-3$};
    \draw[purple!90!black,->] (2.2,4.7) to[bend left=12] (2.7,3.9);
\end{tikzpicture}
\caption{The convex hull records the joint feasibility of the local moment vector.  The box is the coordinate-wise relaxation.  The marked orange point is box-feasible but not hull-feasible.}
\label{fig:delta4-hull}
\end{figure}
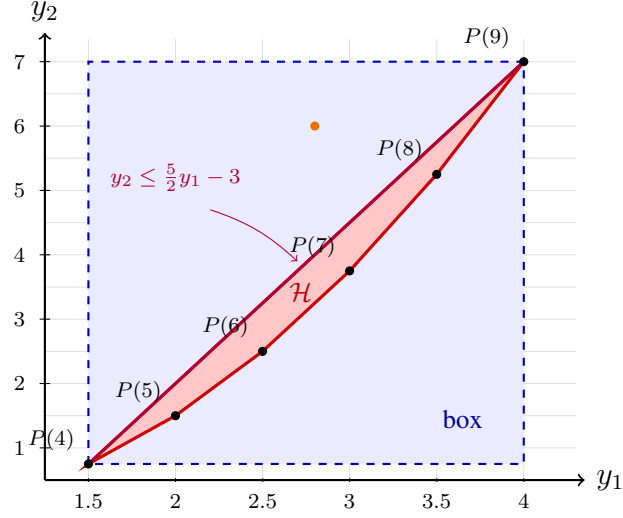
\end{example}

\section{Comparison with Prior LP Bounds}\label{sec:comparison}
A linear programming bound for linear $(r,\delta)$-LRCs was recently established in \cite[Theorem~4.9]{gruicaLRCSDuality2026} (see also \cite{gruicaDualityLP2023}). The construction uses a refined weight distribution $W^{\{j\}}_i(C^\perp)$, the number of weight-$i$ codewords of $C^\perp$ whose support contains the coordinate $j$, and runs an LP over the $n^2$ variables $a_{ij} \triangleq W^{\{j\}}_i(C^\perp)$ for $1 \le i \le n$ and $1 \le j \le n$; see \cite[Notation~4.8]{gruicaLRCSDuality2026}. The constraints are:
\begin{itemize}[noitemsep, topsep=2pt]
\item[(i)] $a_{ij} \ge 0$ for $1 \le i, j \le n$;
\item[(ii)] $a^\perp_{ij} \ge 0$ for $1 \le i, j \le n$;
\item[(iii)] $a^\perp_{ij} = 0$ for $1 \le i \le d-1$ and $1 \le j \le n$;
\item[(iv)] $\sum_{i=1}^{r+\delta-1} a_{ij} \ge q^{\delta-1} - q^{\delta-2}$ for $1 \le j \le n$;
\item[(v)] $a_{1j} = 0$ for $1 \le j \le n$.
\end{itemize}
Here $a^\perp_{il}$ is a specific linear combination of the $a_{js}$'s arising from a refined MacWilliams identity (see~\cite[Notation~4.8]{gruicaLRCSDuality2026}). The objective is to minimize $\sum_{i=1}^n \sum_{j=1}^n a_{ij}/i$. For the code-induced feasible point
\[
a_{ij}=W_i^{\{j\}}(C^\perp),
\]
Lemma~4.7 of \cite{gruicaLRCSDuality2026} gives
\[
\sum_{i=1}^n\sum_{j=1}^n\frac{a_{ij}}{i}
=
|C^\perp|-1.
\]
Consequently, if \(\mu^*\) denotes the minimum of the
refined-weight LP, then
\[
\mu^*\le |C^\perp|-1,
\]
and hence
\[
k\le n-\left\lceil\log_q(1+\mu^*)\right\rceil.
\]

We show that the LP of \cite[Theorem~4.9]{gruicaLRCSDuality2026} with the nondegeneracy assumption is equivalent, after coordinate symmetrization, to our balanced base LP \eqref{lp:primal_dist}--\eqref{lp:coupling}. Thus the new strengthening in the present paper comes from the moments beyond the zeroth one.

\begin{lemma}\label{lem:sym-feasible}
Let \(a=\{a_{ij}\}_{1\le i,j\le n}\) be a feasible point of
the LP in
\cite[Theorem~4.9]{gruicaLRCSDuality2026}.  Define
\[
\bar a_i\triangleq\frac1n\sum_{s=1}^n a_{is},
\qquad
a^{\mathrm{sym}}_{ij}\triangleq\bar a_i,
\qquad
1\le i,j\le n.
\]
Then \(a^{\mathrm{sym}}=\left\{a^{\mathrm{sym}}_{ij}\right\}_{1 \le i,j \le n}\) is feasible and has the same objective
value as \(a\).
\end{lemma}

\begin{proof}
For \(\sigma\in S_n\), define
\[
(\sigma\cdot a)_{ij}
\triangleq
a_{i,\sigma^{-1}(j)}.
\]
The objective and constraints \emph{(i)}, \emph{(iv)}, and
\emph{(v)} of
\cite[Theorem~4.9]{gruicaLRCSDuality2026}
are invariant under this action.

It remains to consider the transformed variables.  Write the
transformation of
\cite[Notation~4.8]{gruicaLRCSDuality2026}
in the form
\[
a^\perp_{i\ell}
=
c_0^{(i)}
+
\sum_{j,s}M^{(i,\ell)}_{j,s}a_{js}.
\]
The coefficient \(M^{(i,\ell)}_{j,s}\) depends on the coordinate
indices \(\ell,s\) only through whether \(\ell=s\).  Hence
\[
M^{(i,\ell)}_{j,\sigma(s)}
=
M^{(i,\sigma^{-1}(\ell))}_{j,s}, \quad \forall \, 1 \le \ell,s \le n. 
\]
It follows that
\begin{align*}
((\sigma\cdot a)^\perp)_{i\ell}
&=
c_0^{(i)}
+
\sum_{j,s}
M^{(i,\ell)}_{j,s}
a_{j,\sigma^{-1}(s)}
\\
&=
c_0^{(i)}
+
\sum_{j,s}
M^{(i,\ell)}_{j,\sigma(s)}a_{js}
\\
&=
c_0^{(i)}
+
\sum_{j,s}
M^{(i,\sigma^{-1}(\ell))}_{j,s}a_{js}
\\
&=
a^\perp_{i,\sigma^{-1}(\ell)}.
\end{align*}
Thus constraints \emph{(ii)} and \emph{(iii)} are also permuted by the action $\sigma$.  Therefore \(\sigma\cdot a\) is feasible
and has the same objective value as \(a\).

Since the feasible region is convex, the group average
\[
\widehat a
\triangleq
\frac1{n!}\sum_{\sigma\in S_n}\sigma\cdot a
\]
is feasible and has the same objective value.  For every
\(i,j\),
\[
\widehat a_{ij}
=
\frac1{n!}
\sum_{\sigma\in S_n}
a_{i,\sigma^{-1}(j)}
=
\frac1n\sum_{s=1}^n a_{is}
=
\bar a_i.
\]
Hence \(\widehat a=a^{\mathrm{sym}}\), proving the result.
\end{proof}

\begin{lemma} \label{lem:radialization}
Suppose that a feasible point of the refined-weight LP is
coordinate-symmetric:
\[
a_{\omega s}=\bar a_\omega,
\qquad
1\le \omega,s\le n.
\]
Define
\[
D_0\triangleq1,
\qquad
D_\omega\triangleq\frac{n}{\omega}\bar a_\omega,
\qquad
1\le\omega\le n.
\]
Then the transformed variables of
\cite[Notation~4.8]{gruicaLRCSDuality2026}
are independent of the coordinate index and satisfy
\begin{align}
\label{eq:radialization}
a^\perp_{i\ell}
=
\frac{i}{n}
\sum_{\omega=0}^n
D_\omega K_i(\omega;n,q),
\qquad
1\le i,\ell\le n.
\end{align}
\end{lemma}

\begin{proof}
For fixed \(\omega,t,\ell\), summing the coefficient of
\(a_{\omega s}\) in
\cite[Notation~4.8]{gruicaLRCSDuality2026}
over \(s\in[n]\) gives
\begin{align*}
&
\frac{n-1}{\omega}
\binom{n-1-\omega}{t-1}
+
\frac{1-\omega}{\omega}
\binom{n-1-\omega}{t-1}
-
\frac1{q-1}\binom{n-\omega}{t-1}
\\
&\qquad
=
\frac{t}{\omega}\binom{n-\omega}{t}
-
\frac1{q-1}\binom{n-\omega}{t-1}.
\end{align*}
Since
\[
\bar a_\omega=\frac{\omega}{n}D_\omega,
\]
the contribution of \(D_\omega\) to
\(a^\perp_{i\ell}\) is \(D_\omega L_{i,\omega}/n\), where
\begin{align*}
L_{i,\omega}
\triangleq
\frac1q
\sum_{t=1}^i
(-1)^{i-t}q^t
\binom{n-t}{i-t}
\left[
(q-1)t\binom{n-\omega}{t}
-
\omega\binom{n-\omega}{t-1}
\right].
\end{align*}
We claim that
\[
L_{i,\omega}=iK_i(\omega;n,q).
\]
Indeed,
\begin{align*}
\sum_{i=1}^nL_{i,\omega}z^{i-1}
&=
\frac1q
\sum_{t=1}^n
q^t
\left[
(q-1)t\binom{n-\omega}{t}
-
\omega\binom{n-\omega}{t-1}
\right]
z^{t-1}(1-z)^{n-t}
\\
&=
(q-1)(n-\omega)
\bigl(1+(q-1)z\bigr)^{n-\omega-1}(1-z)^\omega
\\
&\qquad
-
\omega
\bigl(1+(q-1)z\bigr)^{n-\omega}(1-z)^{\omega-1}
\\
&=
\frac{\mathrm d}{\mathrm dz}
\left[
\bigl(1+(q-1)z\bigr)^{n-\omega}
(1-z)^\omega
\right]
\\
&=
\sum_{i=1}^n
iK_i(\omega;n,q)z^{i-1}.
\end{align*}
Comparing coefficients proves the claim.

Finally,
\[
\binom{n-1}{i-1}(q-1)^i
=
\frac{i}{n}\binom{n}{i}(q-1)^i
=
\frac{i}{n}K_i(0;n,q),
\]
which is precisely the \(\omega=0\) term in
\eqref{eq:radialization}.
\end{proof}
\begin{remark}
Notation~4.8 of \cite{gruicaLRCSDuality2026} writes the lower
summation limit as \(d^\perp\), although \(d^\perp\) is not an
input in Theorem~4.9.  We use the nondegenerate assumption 
\(\ud(C^\perp) \ge 2\); equivalently, the sum may begin at \(1\), since
constraint~(v) imposes \(a_{1s}=0\).
\end{remark}

\begin{proposition}
\label{prop:sym-refined-base}
Let \(\mu^*\) be the optimum of the refined-weight LP in
\cite[Theorem~4.9]{gruicaLRCSDuality2026}, and let
\(U_{\mathrm{base}}\) be the optimum of 
\eqref{lp:primal_dist}--\eqref{lp:coupling}, together with the balanced constraint \(v_1=0\), but without the
positive-order moment constraints
\eqref{lp:box} and \eqref{lp:convexhull}.
Then
\begin{align}
\label{eq:optimum-reciprocity}
U_{\mathrm{base}}
=
\frac{q^n}{1+\mu^*}.
\end{align}
Consequently, the two LPs give the same upper bound on the
dimension of a nondegenerate linear code:
\begin{align}
\label{eq:dimension-equivalence}
\left\lfloor\log_q U_{\mathrm{base}}\right\rfloor
=
n-\left\lceil\log_q(1+\mu^*)\right\rceil.
\end{align}
\end{proposition}

\begin{proof}
By Lemma~\ref{lem:sym-feasible}, the refined-weight LP has
an optimal coordinate-symmetric solution.  We first map any
such feasible point to a feasible point of the balanced base LP.

Let \(a_{\omega s}=\bar a_\omega\) be a symmetric feasible
point of the refined-weight LP, with objective value
\[
\mu
=
n\sum_{\omega=1}^n\frac{\bar a_\omega}{\omega}.
\]
Define
\[
D_0\triangleq1,
\qquad
D_\omega\triangleq\frac{n}{\omega}\bar a_\omega
\quad(1\le\omega\le n),
\qquad
Q\triangleq\sum_{\omega=0}^nD_\omega.
\]
Then
\[
Q=1+\mu.
\]
Set
\[
M\triangleq\frac{q^n}{Q}
\]
and define
\begin{align*}
A_j
&\triangleq
\frac1Q
\sum_{\omega=0}^n
D_\omega K_j(\omega;n,q),
\qquad 0\le j\le n,
\\
v_\omega
&\triangleq
M\bar a_\omega,
\qquad 1\le\omega\le R.
\end{align*}

Since \(K_0(\omega;n,q)=1\), we have \(A_0=1\).  Moreover,
using~\eqref{eq:krawtchouk-sum},
\[
\sum_{j=0}^nA_j
=
\frac1Q
\sum_{\omega=0}^n
D_\omega
\sum_{j=0}^nK_j(\omega;n,q)
=
\frac{q^n}{Q}
=
M.
\]
By Lemma~\ref{lem:radialization},
\[
a^\perp_{j\ell}
=
\frac{j}{n}
\sum_{\omega=0}^n
D_\omega K_j(\omega;n,q)
=
\frac{jQ}{n}A_j.
\]
Consequently, constraints \emph{(ii)} and \emph{(iii)} of the
refined-weight LP imply
\[
A_j\ge0,
\qquad
A_j=0\quad(1\le j\le d-1).
\]

Krawtchouk inversion gives
\begin{align}
\label{eq:radial-transform}
\sum_{j=0}^nA_jK_\omega(j;n,q)
=
\frac{q^n}{Q}D_\omega
=
MD_\omega,
\qquad
0\le\omega\le n.
\end{align}
In particular, the Delsarte inequalities hold.  For
\(1\le\omega\le R\),
\[
MD_\omega
=
M\frac{n}{\omega}\bar a_\omega
=
\frac{n}{\omega}v_\omega,
\]
so the coupling constraints hold with equality.

Constraint \emph{(iv)} of the refined-weight LP gives
\[
\sum_{\omega=1}^R\bar a_\omega
\ge
(1-q^{-1})q^{\delta-1}.
\]
Multiplying by \(M\) gives
\[
\sum_{\omega=1}^Rv_\omega
\ge
(1-q^{-1})q^{\delta-1}M.
\]
Finally, constraint \emph{(v)} gives
\[
D_1=0,
\qquad
v_1=0,
\]
and~\eqref{eq:radial-transform} then gives
\[
\sum_{j=0}^nA_jK_1(j;n,q)=0.
\]
Thus \((A,v)\) is feasible for the balanced base LP, with
objective value
\[
M=\frac{q^n}{1+\mu}.
\]
Taking a refined-weight optimum gives
\begin{align}
\label{eq:first-opt-ineq}
U_{\mathrm{base}}
\ge
\frac{q^n}{1+\mu^*}.
\end{align}

Conversely, let \((A,v)\) be any feasible point of the
balanced base LP and put
\[
M\triangleq\sum_{j=0}^nA_j.
\]
Define
\begin{align*}
D_\omega
&\triangleq
\frac1M
\sum_{j=0}^n
A_jK_\omega(j;n,q),
\qquad 0\le\omega\le n,
\\
\bar a_\omega
&\triangleq
\frac{\omega}{n}D_\omega,
\qquad
a_{\omega s}\triangleq\bar a_\omega,
\qquad
1\le\omega,s\le n.
\end{align*}
The Delsarte constraints imply \(D_\omega\ge0\), and
\(D_0=1\).  Furthermore,
\[
\sum_{\omega=0}^nD_\omega
=
\frac1M
\sum_{j=0}^nA_j
\sum_{\omega=0}^nK_\omega(j;n,q)
=
\frac{q^n}{M},
\]
because \(A_0=1\).

By Lemma~\ref{lem:radialization} and Krawtchouk inversion,
\[
a^\perp_{i\ell}
=
\frac{i}{n}
\sum_{\omega=0}^n
D_\omega K_i(\omega;n,q)
=
\frac{i}{n}\frac{q^n}{M}A_i.
\]
Thus constraints \emph{(ii)} and \emph{(iii)} of the
refined-weight LP follow from the nonnegativity and
minimum-distance constraints on \(A_i\).  Constraint
\emph{(i)} follows from \(D_\omega\ge0\), and balancedness
gives \(D_1=0\), hence constraint \emph{(v)}.

For \(1\le\omega\le R\), the base coupling constraint gives
\[
MD_\omega
=
\sum_{j=0}^nA_jK_\omega(j;n,q)
\ge
\frac{n}{\omega}v_\omega.
\]
Therefore
\[
\bar a_\omega
=
\frac{\omega}{n}D_\omega
\ge
\frac{v_\omega}{M}.
\]
Using the base zeroth-moment constraint, we obtain
\[
\sum_{\omega=1}^R\bar a_\omega
\ge
\frac1M\sum_{\omega=1}^Rv_\omega
\ge
(1-q^{-1})q^{\delta-1},
\]
which is precisely constraint \emph{(iv)}.

The refined-weight objective of this feasible point is
\[
n\sum_{\omega=1}^n
\frac{\bar a_\omega}{\omega}
=
\sum_{\omega=1}^nD_\omega
=
\frac{q^n}{M}-1.
\]
Hence
\[
\mu^*
\le
\frac{q^n}{M}-1,
\]
or equivalently
\[
M
\le
\frac{q^n}{1+\mu^*}.
\]
Since this holds for every feasible base-LP point,
\begin{align}
\label{eq:second-opt-ineq}
U_{\mathrm{base}}
\le
\frac{q^n}{1+\mu^*}.
\end{align}
Combining~\eqref{eq:first-opt-ineq} and
\eqref{eq:second-opt-ineq} proves
\eqref{eq:optimum-reciprocity}.

Finally, since \(n\) is an integer,
\[
\left\lfloor
\log_q\frac{q^n}{1+\mu^*}
\right\rfloor
=
\left\lfloor
n-\log_q(1+\mu^*)
\right\rfloor
=
n-\left\lceil\log_q(1+\mu^*)\right\rceil,
\]
which proves~\eqref{eq:dimension-equivalence}.
\end{proof}

\section{Computational Results}\label{sec:computations}
We implemented the LPs of Section~\ref{sec:LP} for both linear and nonlinear codes using \texttt{PuLP}~3.3.0 in Python~3.11.15.  All LPs were solved with the CBC solver bundled with PuLP, namely CBC~2.10.3 (build date: Dec.~15, 2019). For the linear experiments, the LP gives an upper bound $U$ on the code size; we report the induced dimension bound $\lfloor \log_q U\rfloor$.  For the nonlinear experiments, we report the upper bound $\log_q U$ directly.  In both cases we compare the box LP and the convex-hull LP with the generalized Singleton bound and with the shortening bound for $\delta \ge 3$. We do not include the Gruica--Jany--Ravagnani LP in the main comparison table, because by Proposition~\ref{prop:sym-refined-base} it is already represented by the base part of our LP. There is no need to consider the case $\delta=2$ because the moment constraints \eqref{lp:box} and \eqref{lp:convexhull} become vacuous, and it was already considered in \cite{gruicaLRCSDuality2026}.

The code and the full scan data are available at 
\begin{center}
\texttt{https://github.com/Lsj0815/LWX-LRC}
\end{center}

Our numerical computation covers
\[
q\in\{2,3\},\quad 5\le n\le 40,\quad \delta\in\{3,4,5,6\},\quad \delta \le d\le n,
\quad 1\le r\le n-\delta+1.
\]
For linear codes we record, for each parameter tuple $(q,n,d,r,\delta)$,
\[
 k_{\mathrm{conv}},\quad k_{\mathrm{box}},\quad k_{\mathrm{SH,LP}},\quad
 k_{\mathrm{SH,exact}},\quad k_{\mathrm{GS}}.
\]
For nonlinear codes we record the analogous upper bounds on $\log_q |C|$:
\[
 M_{\mathrm{conv}}^{\mathrm{bal}},\quad M_{\mathrm{conv}}^{\mathrm{unbal}},\quad
 M_{\mathrm{box}}^{\mathrm{bal}},\quad M_{\mathrm{box}}^{\mathrm{unbal}},\quad
 M_{\mathrm{SH,LP}},\quad M_{\mathrm{SH,exact}},\quad k_{\mathrm{GS}}.
\]
Note that \(k_{\mathrm{GS}}\) is an integer bound valid for nonlinear
codes as well, satisfying \(K\le\lceil K\rceil\le
k_{\mathrm{GS}}\). Here ``balanced'' means $v_1=0$ in \eqref{lp:1-moment} (see Remark \ref{7:balanced}); this is automatic for nondegenerate linear codes but is an additional assumption for nonlinear codes.  A strict improvement in Tables~\ref{tab:linear-conv-examples} and~\ref{tab:nonlinear-conv-examples} means that the reported convex-hull bound is smaller than each listed baseline after applying the same rounding convention used for that column.

Tables~\ref{tab:linear-conv-examples}--\ref{tab:regime-summary} summarize the computation without printing the full data set.  The first two tables list representative parameter sets for which our convex-hull bounds are strictly tighter than the shortening LP bound, the table-based shortening bound, and the generalized Singleton bound.  Table~\ref{tab:regime-summary} then records, for four regimes for $d/n$ and $(r+\delta-1)/n$, the total number of parameter sets and the number for which the convex-hull bound is at least as strong as the shortening LP bound.

\begin{table}[H]
\centering
\scriptsize
\setlength{\tabcolsep}{4pt}
\caption{Representative linear-code examples where Conv, the convex-hull LP dimension bound, is strictly smaller than SH-LP, the shortening bound using ordinary Delsarte LP auxiliary estimates; SH-exact, the shortening bound using Grassl table values; and GS, the generalized Singleton bound.}
\label{tab:linear-conv-examples}
\begin{tabular}{ccccc|rrrr}
\toprule
$q$ & $n$ & $d$ & $r$ & $\delta$ & Conv & SH-LP & SH-exact & GS \\
\midrule
2 & 40 & 17 & 1 & 3 & 5 & 6 & 6 & 8 \\
2 & 40 & 16 & 1 & 4 & 5 & 6 & 6 & 7 \\
2 & 40 & 17 & 2 & 5 & 5 & 6 & 6 & 8 \\
2 & 40 & 16 & 4 & 6 & 8 & 9 & 9 & 12 \\
3 & 40 & 23 & 1 & 3 & 4 & 5 & 5 & 6 \\
3 & 12 & 6 & 6 & 4 & 5 & 6 & 6 & 6 \\
3 & 12 & 6 & 6 & 5 & 5 & 6 & 6 & 6 \\
3 & 29 & 15 & 2 & 6 & 4 & 5 & 5 & 5 \\
\bottomrule
\end{tabular}

\end{table}
\begin{table}[H]
\centering
\scriptsize
\setlength{\tabcolsep}{4pt}
\caption{Representative nonlinear-code examples where Conv, the unbalanced convex-hull LP bound, is strictly smaller than SH-LP, the shortening bound using ordinary Delsarte LP auxiliary estimates; SH-exact, the shortening bound using Brouwer table values; and GS, the generalized Singleton bound.  The entries are upper bounds on $\log_q|C|$. The GS column reports the integer bound \(k_{\mathrm{GS}}\), which is
valid for nonlinear codes by~\cite{westerbackPolymatroid2015}.}
\label{tab:nonlinear-conv-examples}
\begin{tabular}{ccccc|rrrr}
\toprule
$q$ & $n$ & $d$ & $r$ & $\delta$ & Conv & SH-LP & SH-exact & GS \\
\midrule
2 & 25 & 10 & 3 & 3 & 7.326 & 8.000 & 8.000 & 10.000 \\
2 & 21 & 8 & 3 & 4 & 6.419 & 7.000 & 7.000 & 8.000 \\
2 & 22 & 8 & 3 & 5 & 6.466 & 7.000 & 7.000 & 7.000 \\
2 & 16 & 6 & 8 & 6 & 7.637 & 8.000 & 8.000 & 8.000 \\
3 & 13 & 3 & 10 & 3 & 9.877 & 10.000 & 10.000 & 10.000 \\
3 & 13 & 4 & 9 & 4 & 8.880 & 9.000 & 8.981 & 9.000 \\
3 & 10 & 5 & 5 & 5 & 4.806 & 5.000 & 5.000 & 5.000 \\
3 & 11 & 6 & 5 & 6 & 4.778 & 5.000 & 5.000 & 5.000 \\
\bottomrule
\end{tabular}
\end{table}

\begin{table}[H]
\centering
\scriptsize
\setlength{\tabcolsep}{4pt}
\caption{Regime summary for the nonlinear exhaustive scan using the unbalanced convex-hull bound.  The column $N$ counts parameter tuples $(q,n,d,r,\delta)$ in each regime, and Conv $\le$ SH-LP counts tuples for which the unbalanced convex-hull bound is no larger than the nonlinear SH-LP bound.  The bins use threshold $0.20$: low means $\le 0.20$, and medium/high means $>0.20$.}
\label{tab:regime-summary}
\begin{tabular}{cc|cc}
\toprule
$d/n$ bin & $(r+\delta-1)/n$ bin & $N$ & Conv $\le$ SH-LP \\
\midrule
low & low & 1212 & 0 \\
low & medium/high & 9952 & 3874 \\
medium/high & low & 9952 & 5617 \\
medium/high & medium/high & 114292 & 96715 \\
\bottomrule
\end{tabular}
\end{table}

The observed pattern is consistent with the structure of the constraints.  The convex-hull LP is never weaker than the box LP because it implies all box inequalities, and it is often strictly stronger when $\delta\ge4$.  Within the binary and ternary scan considered here, the representative examples above show improvements over the shortening LP, the table-based shortening bound, and the generalized Singleton bound.  At the same time, Table~\ref{tab:regime-summary} shows that shortening bounds remain important, especially when both $d/n$ and $(r+\delta-1)/n$ are small.  The new LPs are therefore best understood as complementary to shortening methods rather than as replacements for them.

\section{Conclusion}\label{sec:conclusion}
We developed a moment-based LP framework for $(r,\delta)$-locally
recoverable codes with two main advantages: it is substantially
smaller and more symmetric than the LP of
\cite{gruicaLRCSDuality2026}, and it applies without linearity
assumptions. The computations reported in
Section~\ref{sec:computations} identify many parameter sets where
the higher-moment constraints improve standard baseline bounds,
while also showing that the shortening bound remains stronger in
certain regimes.

Several questions remain open. First, the present LP uses moment
information only up to order $\delta - 2$; it would be interesting
to determine whether additional
structural information can yield further valid
inequalities. Second, the regime in which the shortening bound dominates suggests that a hybrid LP-shortening method may give the stronger practical bounds; formalizing such a hybrid is a natural
next step. Finally, the nonlinear balanced and unbalanced variants of our LP deserve a more systematic comparison with the best available bounds for $A_q(n, d)$, particularly in regimes where Delsarte-type bounds are known to be tight or near-tight. 



\section*{Acknowledgment}
The authors are grateful to Anina Gruica for kindly sharing the SageMath implementation of the LP bound in~\cite{gruicaLRCSDuality2026} and for helpful correspondence.

\bibliographystyle{IEEEtranS}


\end{document}